\pdfoutput=1

\RequirePackage{luatex85}

\documentclass[a4paper,11pt]{article}
\usepackage{iftex}

\ifluatex
\else
\ifxetex
\else
\usepackage[utf8]{inputenc}
\usepackage[T1]{fontenc}
\fi
\fi

\usepackage{amsmath}
\usepackage{amssymb}
\usepackage{amsthm}
\usepackage{mathtools}
\usepackage{thmtools}
\usepackage{enumitem}

\ifluatex
\usepackage{fontspec}
\usepackage{unicode-math}
\else
\ifxetex
\usepackage{fontspec}
\usepackage{unicode-math}
\else
\usepackage{isomath}
\usepackage{libertine}
\usepackage[libertine]{newtxmath} 
\usepackage[scaled=0.96]{zi4} 
\fi
\fi

\usepackage[DIV=11]{typearea} 

\usepackage{microtype}

\usepackage[font=small]{caption}
\usepackage[algo2e,ruled,vlined,linesnumbered]{algorithm2e}

\usepackage{xcolor}

\usepackage[textwidth=2cm,textsize=footnotesize]{todonotes}
\usepackage{lineno}

\usepackage[
	style=alphabetic,
    maxalphanames=4,
    minalphanames=4,
    maxnames=20,
    minnames=20,
	backref=true,
	doi=true,
	url=true,
	backend=biber
]{biblatex}

\ifpdf
\usepackage[ocgcolorlinks]{hyperref} 
\else
\usepackage[hidelinks]{hyperref}
\fi

\allowdisplaybreaks[1]

\makeatletter
\g@addto@macro\bfseries{\boldmath}
\makeatother

\makeatletter
\g@addto@macro\mdseries{\unboldmath}
\g@addto@macro\normalfont{\unboldmath}
\g@addto@macro\rmfamily{\unboldmath}
\g@addto@macro\upshape{\unboldmath}
\makeatother

\makeatletter
\g@addto@macro\bfseries{\boldmath}
\def\thmhead@plain#1#2#3{%
  \thmname{#1}\thmnumber{\@ifnotempty{#1}{ }\@upn{#2}}%
  \thmnote{ {\the\thm@notefont\unboldmath(#3)}}}
\let\thmhead\thmhead@plain
\makeatother

\DeclareCiteCommand{\citem}
    {}
    {\mkbibbrackets{\bibhyperref{\usebibmacro{postnote}}}}
    {\multicitedelim}
    {}

\renewcommand*{\multicitedelim}{\addcomma\space}

\AtEveryCitekey{\clearfield{note}}

\makeatletter
\AtBeginDocument{%
    \newlength{\temp@x}%
    \newlength{\temp@y}%
    \newlength{\temp@w}%
    \newlength{\temp@h}%
    \def\my@coords#1#2#3#4{%
      \setlength{\temp@x}{#1}%
      \setlength{\temp@y}{#2}%
      \setlength{\temp@w}{#3}%
      \setlength{\temp@h}{#4}%
      \adjustlengths{}%
      \my@pdfliteral{\strip@pt\temp@x\space\strip@pt\temp@y\space\strip@pt\temp@w\space\strip@pt\temp@h\space re}}%
    \ifpdf
      \typeout{In PDF mode}%
      \def\my@pdfliteral#1{\pdfliteral page{#1}}
      \def\adjustlengths{}%
    \fi
    \ifxetex
      \def\my@pdfliteral #1{}
      \def\adjustlengths{\setlength{\temp@h}{-\temp@h}\addtolength{\temp@y}{1in}\addtolength{\temp@x}{-1in}}%
    \fi%
    \def\Hy@colorlink#1{%
      \begingroup
        \ifHy@ocgcolorlinks
          \def\Hy@ocgcolor{#1}%
          \my@pdfliteral{q}%
          \my@pdfliteral{7 Tr}
        \else
          \HyColor@UseColor#1%
        \fi
    }%
    \def\Hy@endcolorlink{%
      \ifHy@ocgcolorlinks%
        \my@pdfliteral{/OC/OCPrint BDC}%
        \my@coords{0pt}{0pt}{\pdfpagewidth}{\pdfpageheight}%
        \my@pdfliteral{F}
        \my@pdfliteral{EMC/OC/OCView BDC}%
        \begingroup%
          \expandafter\HyColor@UseColor\Hy@ocgcolor%
          \my@coords{0pt}{0pt}{\pdfpagewidth}{\pdfpageheight}%
          \my@pdfliteral{F}
        \endgroup%
        \my@pdfliteral{EMC}%
        \my@pdfliteral{0 Tr}
        \my@pdfliteral{Q}%
      \fi
      \endgroup
    }%
}
\makeatother

\colorlet{DarkRed}{red!50!black}
\colorlet{DarkGreen}{green!50!black}
\colorlet{DarkBlue}{blue!50!black}

\hypersetup{
	linkcolor = DarkRed,
	citecolor = DarkGreen,
	urlcolor = DarkBlue,
	bookmarksnumbered = true,
	linktocpage = true
}

\DontPrintSemicolon
\SetKwProg{Procedure}{Procedure}{}{}
\SetProcNameSty{textsc}
\SetFuncSty{textsc}
\SetKw{KwAnd}{and}
\SetKw{KwDo}{do}

\SetCommentSty{mycommfont}

\declaretheorem[numberwithin=section]{theorem}
\declaretheorem[numberlike=theorem]{lemma}

\declaretheorem[numberlike=theorem]{corollary}
\declaretheorem[numberlike=theorem]{definition}
\declaretheorem[numberlike=theorem]{claim}
\declaretheorem[numberlike=theorem]{observation}

\newtheoremstyle{casestyle}  
  {3pt}   
  {3pt}   
  {}      
  {}      
  {}      
  {:}     
  { }     
  {\itshape #1~#2}    

\theoremstyle{casestyle}
\newtheorem{case}{Case}[] 
\newtheorem{subcase}{Case}[case]

\AtBeginEnvironment{proof}{
  \setcounter{case}{0}%
  \setcounter{subcase}{0}%
  \setcounter{subsubcase}{0}%
}

\usetikzlibrary{decorations.pathreplacing,decorations.pathmorphing,calc}
\usepackage{float}

\newcommand{\bigO}{O}

\newcommand{\kl}{k(l)}
\newcommand{\kll}{k(l+1)}

\DeclareMathOperator{\w}{w}
\DeclareMathOperator{\dist}{d}
\DeclareMathOperator{\edgesop}{E}
\DeclareMathOperator{\origin}{origin} 

\newcommand{\ceil}[1]{\left\lceil #1 \right\rceil}

\numberwithin{equation}{subsection}

\title{A General Reduction from Near-Additive Emulators to Near-Exact Hopsets\thanks{This project has received funding from the European Research Council (ERC) under the European Union's Horizon 2020 research and innovation programme (grant agreement No 947702).}}

\author{
    Julian Aeri \thanks{Department of Computer Science, University of Salzburg, Austria}
    \and
    Sebastian Forster \footnotemark[2]
    \and
    Mara Grilnberger  \footnotemark[2] \thanks{This publication has been supported by the EXDIGIT (Excellence in Digital Sciences and Interdisciplinary Technologies) project, funded by Land Salzburg under grant number 20204-WISS/263/6-6022.}
}

\date{}

\ifpdf
\hypersetup{
	pdftitle = {A General Reduction from Near-Additive Emulators to Near-Exact Hopsets},
	pdfauthor = {Julian Aeri, Sebastian Forster, Mara Grilnberger}
}
\else
\fi

\begin{document}
\maketitle
\begin{abstract}
Graph emulators and hopsets are two fundamental concepts for distance approximation.
For a given graph $G$, an $(\alpha,\beta)$-emulator is a sparse graph on the same vertex set 
that preserves the distances of $G$ up to a multiplicative stretch $\alpha$ and additive stretch $\beta$. 
In contrast, an $(\alpha,\beta)$-hopset is a set of additional edges that, when added to $G$, ensures that distances can be approximated up to a multiplicative stretch $\alpha$, using paths containing at most $\beta$ edges. 
When $\alpha = 1+\epsilon$ for arbitrarily small $\epsilon>0$, these structures are known as near-additive emulators and near-exact hopsets, respectively. 
Prior work showed that there is a remarkable similarity between the constructions and guarantees of these two objects. 
In their survey on this topic, Elkin and Neiman [Bull. EATCS 130, 2020] explicitly asked whether one can obtain a general reduction between near-additive emulators and near-exact hopsets.
Following that, Kogan and Parter [FOCS, 2022] provided a general reduction from hopsets to emulators and spanners.

In this paper, we address the reverse direction and show that any construction for a near-additive emulator for undirected unweighted graphs can be leveraged as a black box to construct a hopset for an undirected weighted graph with comparable size, stretch, and a hopbound comparable to the emulator's additive stretch. 
Specifically, we show that any algorithm that constructs a $(1+\epsilon',\beta)$-emulator, with $0 \le \epsilon' \le 1$ and $\beta \ge 1$, of size $S_{\mathcal{A}}(n, \epsilon',\beta)$, can be used to obtain a $(1+\epsilon, \bigO(\frac{\beta^2}{\epsilon^2} \ln(\frac{n}{\epsilon})))$-hopset of size $\bigO((S_{\mathcal{A}}(n+m\frac{\beta}{\epsilon^2}, \frac{\epsilon}{294},\beta) \frac{1}{\epsilon} + n)\ln(\frac{n}{\epsilon}))$, for any $0 < \epsilon \le 1$.
Therefore, our reduction answers the question of Elkin and Neiman [Bull. EATCS 130, 2020] for sparse graphs and further advances the understanding of the formal connection between these two structures. Designing a reduction resulting in a hopset size that does not depend on $m$ remains an intriguing open question.

\end{abstract}

\newpage
\tableofcontents
\newpage

\section{Introduction}
Given a graph $G$, a hopset is a set of additional edges that, when added to $G$, ensures that for any pair of vertices, an approximate shortest path exists with a bounded number of edges (or hops). 
They were first formally introduced by Cohen \cite{Coh94} as a tool for efficient parallel computation of approximate shortest paths.  And since then they have been central in parallel algorithms \cite{Coh94, MVPX15, EN19, Fin18, CFR20a, CFR20b, CF23}, dynamic distance maintenance \cite{Ber09, HKN18, GW20, LN22}, and distributed shortest paths computations \cite{HKN16, EN19, CDKL21}.
\begin{definition}[$h$-hop distance]
    Given a graph $G = (V,E)$.
    For a pair of vertices $u,v \in V$, the $h$-hop distance $\dist_{G}^{h}(u,v)$ denotes the weight of the shortest path from $u$ to $v$ in $G$ that contains at most $h$ edges.
\end{definition}
\begin{definition}[$(\alpha, \beta)$-hopset]
    Given a graph $G = (V,E)$, a multiplicative stretch $\alpha\ge1$ and a hopbound $\beta \in \mathbb{N}$, a $(\alpha, \beta)$-hopset $F \subseteq V \times V$ is a set of weighted edges, such that
    \begin{equation*} 
        \forall{u,v \in V} \quad \dist_G(u,v) \le \dist_{H}^{\beta}(u,v) \le \alpha \cdot \dist_{G}(u,v)
    \end{equation*}
    where $H = (V, E \cup F)$ is the hopset graph.
    A hopset is called exact if $\alpha =1$, and near-exact if $\alpha = 1+\epsilon$ for arbitrarily small $\epsilon > 0$.
\end{definition}

Two other fundamental structures for approximating distances are spanners and emulators.
Rather than bounding the number of hops, they instead aim to sparsify the graph, while preserving approximate pairwise distances. Formally, an emulator is defined as follows:
\begin{definition}[$(\alpha,\beta)$-emulator]
    Given an undirected graph $G = (V,E)$, a multiplicative stretch $\alpha \ge 1$ and an additive stretch $\beta \ge 0$, the graph $G^* = (V, E^*)$ with $E^* \subseteq V \times V$ is an $(\alpha,\beta)$-emulator of $G$, iff
    \begin{equation*}
        \forall_{u,v \in V} \quad \dist_G(u,v) \le \dist_{G^*}(u,v) \le \alpha \cdot \dist_{G}(u,v) + \beta.
    \end{equation*}
    If in addition $E^* \subseteq E$, then $G^*$ is a spanner of $G$.
\end{definition}
Peleg and Schäffer first studied multiplicative spanners (where $\beta = 0$) in \cite{PS89}, and additive emulators (where $\alpha = 1$) were introduced by Dor, Halperin, and Zwick \cite{DHZ00}.
The concept of near-additive spanners and emulators was later introduced by Elkin and Peleg \cite{EP04}, who showed that one can construct spanners with stretch $1+\epsilon$ for arbitrarily small $\epsilon > 0$ and near-linear size, at the cost of allowing a sufficiently large constant additive stretch. In particular, they showed that for every integer $k \ge 1$, there exists a $(1+\epsilon, \beta)$-spanner with $\bigO(\beta \cdot n^{1+1/k})$ edges, where $\beta = \bigO(\log(k) / \epsilon)^{\log(k)}$.
Thorup and Zwick \cite{TZ06} later gave different constructions for near-additive spanners and emulators on undirected, unweighted graphs. A notable property of their construction is that it is universal, i.e., the algorithm is independent of $\epsilon$, and thus the resulting spanner applies for all $\epsilon > 0$ simultaneously.
Abboud, Bodwin and Pettie \cite{ABP18} provided a lower bound on the trade-off between $\epsilon$, the additive stretch and the sparsity of emulators.

In recent years, also hopsets for directed graphs have been studied intensively \cite{KP22b, BW23, BH23}.
Kogan and Parter \cite{KP22b} showed that for any directed graph, there is a near-exact hopset with size $\bigO(n)$ and a hopbound of $\tilde{\bigO}(n^{2/5})$, improving upon a long-standing folklore sampling algorithm, that implied a near-exact hopset of size $\bigO(n)$ with a hopbound of $\tilde{\bigO}(n^{1/2})$.  Bernstein and Wein \cite{BW23} further improved on this result and showed that one can obtain a linear-sized near-exact hopset with a hopbound of $\tilde{\bigO}(n^{1/3})$. Bodwin and Hoppenworth \cite{BH23} showed that the folklore sampling is essentially optimal for exact-hopsets in both directed and undirected graphs. 

For undirected graphs, there is an interesting and extensively studied relationship between the constructions of near-additive spanners, emulators, and near-exact hopsets ~\cite{EN19, EN17, ABP18, HP19, EN19b, BenLevyP20, NS22}.
Around the same time, Elkin and Neiman \cite{EN19b} and Huang and Pettie \cite{HP19} independently showed that classic constructions for emulators based on \cite{TZ06} can be used to obtain near-exact hopsets with a constant hopbound. In particular for any integer $k \ge 1$, Elkin and Neiman \cite{EN19b} gave a construction, that yields a near-exact hopset with hopbound $\bigO(k / \epsilon)^{k}$ and size $\bigO(n^{1+1/(2^k -1)})$, while Huang and Pettie \cite{HP19} showed a construction yielding a near-exact hopset with hopbound $\bigO(k /\epsilon)^k$ and size $\bigO(n^{1+1/(2^{k + 1}-1)})$.
These results were unified by Neiman and Shabat \cite{NS22}, who devised a single algorithm that can provide state-of-the-art hopsets for undirected, weighted graphs with various stretch regimes.

Elkin and Neiman \cite{EN20} gave an extensive survey on the connection between these objects, where they raised the following question:
\begin{quote} \cite{EN20}
    \textit{"A very interesting open problem is to explain the relationship 
    between near-additive spanners and near-exact hopsets rigorously, i.e., 
    by providing a reduction between these two objects."}
\end{quote}
Subsequently, Kogan and Parter \cite{KP22a} made significant progress on this question by providing a general reduction from hopsets to emulators, spanners, and distance preservers. While their reduction to distance preservers is a bit more involved, the reduction to emulators and spanners is quite simple. 
\begin{observation}[\cite{KP22a}]
    Let $G = (V,E)$ be an unweighted $n$-vertex Graph. Let $H^*$ be some $(1+\epsilon,\beta)$- hopset for $G$ and let $G^*$ be a multiplicative spanner with stretch $t$ for $G$. Then, $H^* \cup G^*$ is a $(1+\epsilon, \beta \cdot t)$-emulator.
\end{observation}This result also extends to spanners, by a standard reduction.
More recently, Kogan and Parter \cite{KP25} initiated progress in the reverse direction, by showing that certain classes of distance preservers can be converted into exact hopsets.

\subsection{Our Results}\label{section:our-result} 
In this work, we present a general reduction from near-additive emulators to near-exact hopsets, thereby answering the question of \cite{EN20} in the affirmative for sparse graphs. As our main contribution, we present a general reduction that shows that any emulator construction can be used to obtain a hopset with similar guarantees, formalized as follows.
\begin{restatable}{theorem}{thmmain}\label{thm:main}
    Let $G = (V,E)$ be an undirected weighted graph, with edge weights in $[1,W]$ and let $0<\epsilon \le 1$.
    Suppose there exists an algorithm $\mathcal{A}$ that, for any undirected unweighted $n$-vertex graph, constructs an $(\alpha, \beta)$-emulator, with $\alpha \ge 1$ and $\beta \ge 0$, of size at most $S_{\mathcal{A}}(n,\alpha,\beta)$.
    Then there exists a $(\alpha(1+24\epsilon), \bigO(\alpha   t^2)\ln(nW))$-hopset of $G$ of size $\bigO(S_{\mathcal{A}}(n+m\frac{t}{\epsilon}, \alpha, \beta) \frac{1}{\epsilon}\ln(nW))$, where $t = \max(\frac{1}{\epsilon}, \frac{\beta}{\epsilon})$.
\end{restatable}

A direct consequence of Theorem~\ref{thm:main}, by setting $\alpha = 1 + \epsilon$, is the following result for near-additive emulators, thus answering the question raised in \cite{EN20} for sparse graphs.

\begin{corollary}\label{cor:main} 
    Let $G = (V,E)$ be an undirected weighted graph, with edge weights in $[1,W]$ and let $0<\epsilon \le 1$.
    Assume there exists an algorithm $\mathcal{A}$ that, for any undirected unweighted $n$-vertex graph and any $0 \le \epsilon' \le 1$, constructs a $(1+\epsilon', \beta)$-emulator, with $\beta \ge 1$, of size at most $S_{\mathcal{A}}(n,\epsilon',\beta)$. Then $G$ admits a $(1+\epsilon, \bigO(\frac{\beta^2}{\epsilon^2})\ln(nW))$-hopset of size $\bigO(S_{\mathcal{A}}(n+m\frac{\beta}{\epsilon^2}, \frac{\epsilon}{49}, \beta)\frac{1}{\epsilon}\ln(nW))$.
\end{corollary}

Finally, we can make use of the reduction from~\cite{EN19} to remove the maximum weight $W$ from the log-factor and get the following overall result.

\begin{corollary}\label{cor:main2} 
    Let $G = (V,E)$ be an undirected weighted graph and let $0<\epsilon \le \frac{1}{2}$.
    Assume there exists an algorithm $\mathcal{A}$ that, for any undirected unweighted $n$-vertex graph and any $0 \le \epsilon' \le \frac{1}{2}$, constructs a $(1+\epsilon', \beta)$-emulator, with $\beta \ge 1$, of size at most $S_{\mathcal{A}}(n,\epsilon', \beta)$. Then $G$ admits a $(1+\epsilon, \bigO(\frac{\beta^2}{\epsilon^2}\ln(\frac{n}{\epsilon})))$-hopset of size $\bigO((S_{\mathcal{A}}(n+m\frac{\beta}{\epsilon^2}, \frac{\epsilon}{294}, \beta)\frac{1}{\epsilon} + n)\ln(\frac{n}{\epsilon}))$.
\end{corollary}

Compared to the ``ideal result'', having guarantees as close to the original emulator as possible, our reduction incurs a blowup in the size depending on the number of edges $m$ in the input graph, the dependence on $\beta$ in the hopbound is quadratic instead of linear, and our method introduces an additional log-factor in both the size and the hopbound.
Although the size of the obtained hopset also depends on $m$, we believe that this has only a minor impact in practice, as the runtime of typical algorithmic applications usually already depends at least linearly on $m$ and therefore adding roughly $m$ additional edges does not change the asymptotic complexity.

\subsection{Overview of Techniques}
First, we introduce the notation and relevant definitions, then we give a high-level overview of the techniques used to achieve the reduction.

\paragraph*{Notation} In this work, we consider an undirected, weighted graph $G=(V,E)$ with positive edge weights, where $V$ is the set of vertices and $E$ is the set of edges, with $n := |V|$ and $m := |E|$. Given an edge $(u,v) \in E$, its weight is denoted by the function $\w_G \colon E\to\mathbb{R}_{\ge 1}$, and all edge weights lie in the range $[1,W]$.
More generally, for any weighted edge set $F$, we denote the weight of an edge $e \in F$ by $\w_F(e)$.
Given a path $\pi$ in $G$, its weight is defined as $\w_G(\pi) = \sum_{(u,v)\in\pi} \w_G((u,v))$ and $|\pi|$ is the number of edges contained in $\pi$. For any two vertices $u,v$ on $\pi$, we denote by $\pi[u,v]$ the subpath of $\pi$ between $u$ and $v$.
For a pair of vertices $u,v \in V$, the distance $\dist_G(u,v)$ is defined as the minimum weight of any path from $u$ to $v$ in $G$. Such a path with minimum weight is referred to as a shortest path.
If no such path exists, we define $\dist_G(u,v) := \infty$.

\paragraph*{Reduction} 
Given a weighted undirected graph $ G = (V, E)$, we organize our construction around distance scales. We consider the scales $[(1+\epsilon)^i, (1+\epsilon)^{i+1})$ for some $0 < \epsilon \le 1$ and $i \in \mathbb{N}_0$. 

The core idea behind our reduction is to apply a given $(\alpha,\beta)$-emulator construction, for unweighted undirected graphs, independently at each scale.
For every scale $i$, we define an unweighted graph $G_i$ that captures the distances in $G$ at that scale. In particular, we connect pairs of vertices whose distance falls in the given range by a single edge. We then construct an emulator $M_i$ on $G_i$. Since the emulators approximate the distances in the unweighted graphs, we can bound the number of hops of the resulting approximate shortest paths, and by appropriately rescaling the weight of these emulator edges, we obtain the approximate distances in the graph $G$.
The hopset $F$ is then formed by taking the union of these rescaled emulator edges over all distance scales. 

To analyze the stretch and hopbound, we show that for any two vertices $u,v \in V$ we can find a sequence of hop-bounded alternative paths, a shortcut, each skipping a subpath of the shortest path $\pi$ from $u$ to $v$ in $G$. Every shortcut reduces the remaining distance to $v$ by a fixed fraction, ensuring that only a logarithmic number of such shortcuts are required. The idea of constructing a sequence of shortcuts follows the strategy used by Bernstein in \cite{Ber09} and later by Henzinger, Krinninger, and Nanongkai in \cite{HKN16}. 

Next, we consider a single such shortcut. We need to ensure that a \textit{sufficiently long} consecutive subpath of the shortcut is contained entirely within an emulator $M_i$ of a single scale. If we concatenate multiple shortcuts from emulators at different scales, the additive stretch $\beta$ incurred by the emulator construction would appear additionally each time the path ``switches'' scale. Hence, preventing the need to make such ``switches'' allows $\beta$ to be absorbed into the multiplicative stretch. To this end, a primary challenge is to find a sequence of vertices, starting from $u$, on the path $\pi$ in $G$, such that:
\begin{itemize}
    \item two consecutive vertices are connected by a single edge in $G_i$ and
    \item the last vertex in the sequence is  \textit{sufficiently far} from $u$.
\end{itemize} Then the corresponding emulator $M_i$ can provide (almost) the entire required shortcut path, and the number of hops is roughly given by $\alpha k + \beta$, if the aforementioned sequence contains $k+1$ points. This is non-trivial, as heavy edges (relative to the scale $i$) on the path may cause consecutive points in the sequence not to be connected in $G_i$. Intuitively, we would require a point of the sequence to lie ``on'' that edge. To address this issue, we subdivide these heavy edges in the graph $G$ at each scale before defining the unweighted graphs $G_i$. We implicitly show that the error incurred by this ``discretization'' of large ``continuous'' edges is negligible. However, this approach introduces additional vertices and, as a result, the hopset $F$ must be formally defined on the graph $G' = (V', E)$, where $V \subseteq V'$. 
It is sufficient to require the ``hopset guarantees'' to hold only for the vertices in $G$. This leads to the following definition.
\begin{definition}[$T$-constrained $(\alpha,\beta)$-hopset]
    Given an undirected graph $G = (V,E)$ and a subset of vertices $T \subseteq V$, for a multiplicative stretch $\alpha\ge1$, and a hopbound $\beta \in \mathbb{N}$, a $T$-constrained $(\alpha, \beta)$-hopset is a set of weighted edges $F \subseteq V \times V$, such that
    \begin{equation*} 
        \forall{u,v \in T} \quad \dist_G(u,v) \le \dist_{H}^{\beta}(u,v) \le \alpha \cdot \dist_{G}(u,v)
    \end{equation*}
    where $H = (V, E \cup F)$ is the hopset graph.
\end{definition}

To obtain a hopset for the original graph $G$, we provide a projection that maps the $V$-constrained hopset $F$ of $G'$ back to a hopset defined entirely on the vertex set of $G$, while preserving the same stretch and hopbound guarantees. 
Conceptually, each subdivision vertex in $V' \setminus V$ represents a specific point ``on'' some original edge in $E$, and the idea behind the projection is to shift hopset edges that are incident to such a subdivision vertex to the nearby original vertices in $V$. 
Consider such an edge $(u',v') \in F$ , suppose $u' \in V'$ lies ``on'' the edge $(u_1, u_2) \in E$ (where possibly $u' \in V$, in which case $u' = u_1 = u_2$), and similarly suppose $v' \in V'$ lies ``on'' the edge $(v_1, v_2) \in E$ (again allowing $v' = v_1 = v_2$, when $v' \in V$). 
We add up to four replacement edges connecting each endpoint $u_1$, $u_2$ to $v_1$ and $v_2$, to retain paths of the same total weight between vertices in $V$ in each rescaled $M_i$ after the removal of the subdivision vertices. Now, we need to ensure that we do not introduce any shortcuts that would lead to underestimating the distances by assigning weight to these new edges accordingly. First, we rename the nodes $u_1$ and $u_2$ such that $u_2$ is the node encountered (first) when traversing the shortest path in $G'$ restricted to the relevant scale from $u'$ to $v'$. We do the same for $v_1$ and $v_2$, where $v_1$ should be encountered (first) when traversing from the path starting from $v'$. For the new edge $(u_i, v_j)$ with $i,j \in \{1,2\}$, the weight is set such that the path $u_1,u', v', v_2$ in $G' \cup F$ has the same total weight as the path $u_1, u_i, v_j, v_2$ (consisting of one to three edges) in $G \cup \{(u_i, v_j)\}$. This process is illustrated in Figure~\ref{fig:projectionexample}.

\begin{figure}[h]
    \centering
        \begin{minipage}[b]{0.45\textwidth}
        \centering
        \begin{tikzpicture}[baseline={(0,0)}, scale=1, every node/.style={font=\small}]
            \useasboundingbox (0,-1.5) rectangle (6,1.5);
        
            \coordinate (up) at (0,0);
            \coordinate (u) at (1,0);
            \coordinate (upp) at (2,0);
            
            \coordinate (vp) at (4,0);
            \coordinate (v) at (5,0);
            \coordinate (vpp) at (6,0);
    
            \draw [thick, gray] (up) -- (upp);
            \draw[thick, gray] (vp) -- (vpp);

            \draw [decorate,decoration={brace,amplitude=7pt,mirror}]
                ([yshift=-.7cm]up.center) -- ([yshift=-.7cm]u.center) node [black,midway,yshift=-0.6cm] 
                { $d_{u_1}$};
            \draw [decorate,decoration={brace,amplitude=7pt,mirror}]
                ([yshift=-.7cm]u.center) -- ([yshift=-.7cm]upp.center) node [black,midway,yshift=-0.6cm] 
                { $d_{u_2}$};
            \draw [decorate,decoration={brace,amplitude=7pt,mirror}]
                ([yshift=-.7cm]vp.center) -- ([yshift=-.7cm]v.center) node [black,midway,yshift=-0.6cm] 
                { $d_{v_1}$};
            \draw [decorate,decoration={brace,amplitude=7pt,mirror}]
                ([yshift=-.7cm]v.center) -- ([yshift=-.7cm]vpp.center) node [black,midway,yshift=-0.6cm] 
                { $d_{v_2}$};

            \draw[red, thick] 
                (u) to [bend left= 60] 
                node [midway, above] {$e \in F$}
                (v) ;

            \draw[thick, dashed, color=violet, decorate, 
              decoration={random steps, segment length=5pt, amplitude=2pt}]
            (upp) .. controls (2.5,-1) and (3.5,-1) .. (vp);
            \draw [thick, violet, dashed] (u) -- (upp);
            \draw[thick, violet, dashed] (vp) -- (v);

            \draw ($(upp) + (1,0)$) node[anchor=mid] {$\dots$};

            \fill (up) circle (1.5pt) node[above=4pt] {$u_1$};
            \fill (upp) circle (1.5pt) node[above=4pt] {$u_2$};

            \fill (vp) circle (1.5pt) node[above=4pt] {$v_1$};
            \fill (vpp) circle (1.5pt) node[above=4pt] {$v_2$};

            \draw[thick] ($(u)+(0,-3pt)$) -- ($(u)+(0,3pt)$) node[below=8pt] {$u'$};
            \draw[thick] ($(v)+(0,-3pt)$) -- ($(v)+(0,3pt)$) node[below=8pt] {$v'$};
        \end{tikzpicture}
        \caption*{(a) $e \in F$ between subdivision vertices}
    \end{minipage}
    \hfill
    \begin{minipage}[b]{0.45\textwidth}
        \centering
        \begin{tikzpicture}[baseline={(0,0)}, scale=1, every node/.style={font=\small}]
            \useasboundingbox (0,-1.5) rectangle (6,1.5);
        
            \coordinate (up) at (0,0);
            \coordinate (u) at (1,0);
            \coordinate (upp) at (2,0);
            
            \coordinate (vp) at (4,0);
            \coordinate (v) at (5,0);
            \coordinate (vpp) at (6,0);
    
            \draw [thick] (up) -- (upp);
            \draw[thick] (vp) -- (vpp);

            \draw ($(upp) + (1,0)$) node[anchor=mid] {$\dots$};

            \draw[blue, thick] (up) to [bend left = 60] 
            node [midway, above] {$d_{u_1} + \w_F(e) - d_{v_1}$}(vp);
            \draw[blue, thick] (up) to [bend right = 45]
            node [midway, below] {$d_{u_1} + \w_F(e) + d_{v_2}$}(vpp);

            \draw[teal, thick] (upp) to [bend left = 25]
            node [midway, above, xshift=-20pt] {$\w_F(e)- d_{u_2} - d_{v_1}$}(vp);
            \draw[teal, thick] (upp) to [bend right = 25]
            node [midway, below, xshift=-20pt] {$\w_F(e)- d_{u_2} + d_{v_2}$}(vpp);

            \fill (up) circle (1.5pt) node[above=4pt] {$u_1$};
            \fill (upp) circle (1.5pt) node[below=4pt] {$u_2$};

            \fill (vp) circle (1.5pt) node[below=4pt] {$v_1$};
            \fill (vpp) circle (1.5pt) node[above=4pt] {$v_2$};

            \draw[thick] ($(u)+(0,-3pt)$) -- ($(u)+(0,3pt)$);
            \draw[thick] ($(v)+(0,-3pt)$) -- ($(v)+(0,3pt)$);

        \end{tikzpicture}
        \caption*{(b) weights of edges replacing $e$}
    \end{minipage}
    \hfill
    \caption{An example for the projection of an edge. The gray lines depict the edges of $G$, original vertices in $V$ are shown as dots, whereas subdivision vertices are marked by a vertical line. The red line is a hopset edge in $F$. The violet path depicts the shortest path in the subdivided version of $G$ from $u'$ to $v'$. The blue and teal edges depict the replacement edges for $e = (u',v')$.}
    \label{fig:projectionexample}
\end{figure}

Originally, the hopset edge $e$ had to have weights such that $\w_F(e)$ was at least as large as the total weight of the shortest path from $u'$ to $v'$ in the subdivided graph of $G$ within one scale. Furthermore, that path contained two edge-disjoint subpaths from $u'$ to $u_2$ and from $v_1$ to $v'$. Therefore, we get that replacing $e$ with four edges between original nodes in $V$, as described above, will not introduce any shortcuts.

\subsubsection*{Roadmap}
The proof of Theorem~\ref{thm:main} is split into two main parts; building a hopset on an extended graph with additional vertices and converting it into a hopset on $G$ while retaining the guarantees on hopbound and stretch. First, we describe the construction of a hopset on the extended graph in Section~\ref{section:construction}. Then, we analyze the resulting size of this hopset in Section~\ref{section:size-analysis}, and in Section~\ref{section:hop-reduction} we consider the resulting stretch and hopbound. In Section~\ref{section:projection}, we provide details of the projection to a hopset on $G$ and prove that the same guarantees hold.

\section{Hopset Reduction}\label{section:hopset-reduction} 

\subsection{Construction}\label{section:construction}
In this section, we describe the construction of the hopset, which is organized around distance scales of the form $[(1 + \epsilon)^i, (1 + \epsilon)^{i+1})$, for an integer $i \ge 0$ and $0<\epsilon\le 1$.
Let $i_\text{max}$ denote the index such that the diameter of the graph falls into the distance scale $[(1+\epsilon)^{i_\text{max}}, (1+\epsilon)^{i_\text{max}+1})$. We therefore only need to consider the indices $0 \le i \le i_{\text{max}}$, as all higher scales are empty. 
Since the diameter of $G$ is bounded by $nW$, we have $i_{\text{max}} \le \lceil \log_{1+\epsilon}(nW) \rceil$ distinct distance scales. 
At each scale $i$, we first create a subdivision graph $\tilde{G}_i$ of $G$ with edge weights bounded relative to the current scale. Over the vertex set of $\tilde{G}_i$, we then construct an unweighted graph $G_i$ which connects vertex pairs whose distance lies within the current scale. To bound the number of hops, we use an arbitrary $(\alpha,\beta)$-emulator of $G_i$, whose weights are rescaled to approximate the distances in $G$. 
The hopset is obtained by taking the union over all distance scales $i$,
of the rescaled emulator edges, together with the auxiliary edges introduced
during the subdivision at each scale.
In order to ensure the desired hopbound, we define $ t := \max(\frac{1}{\epsilon}, \frac{\beta}{\epsilon})$. 
\paragraph*{Subdivision Step} 
Let $\tilde{G}_i := (\tilde{V}_i, \tilde{E}_i)$ be the graph, obtained from $G$ by subdividing each edge $e \in E$ whose weight satisfies
\begin{equation}\label{equa:subdivisionCondition}
   \epsilon(1+\epsilon)^i < \w_G(e) \le t (1+\epsilon)^{i+1} 
\end{equation}
into $k$-segments, with $k = \ceil{\frac{\w_G(e)}{\epsilon(1+\epsilon)^i}}$. \\
The edge $e$ is replaced by a path of $k$ edges $e_1,\dots,e_k$, each with weight $\w_{\tilde{G}_i}(e_j) := \frac{\w_G(e)}{k}$ for $1 \le j \le k$, so that the total weight of the path remains $\w_G(e)$. 
Thus, for the edge $e = (v,w) =(u_0,u_k)$, we introduce $k-1$ new vertices $u_1, \dots, u_{k-1}$, so that $e_j = (u_{j-1}, u_j)$.

Based on this subdivision, we now define the notion of a vertex's origin. 
\begin{restatable}{definition}{deforigin}\label{def:origin}
    For any vertex $u \in \tilde{V}_i$, the origin of $u$ is defined as:
    \begin{equation*}
    \origin(u) :=
        \begin{cases}
            \{u\} & \text{if } u\in V\\[4pt]
            \{v,w\} & \text{if $u \in \tilde{V}_i \setminus V$ and $u$ was created in the subdivision of the edge $(v,w)\in E$.}
        \end{cases}
    \end{equation*}
\end{restatable}
\begin{restatable}{definition}{defdelta}\label{def:delta}
    For any vertex $u \in \tilde{V}_i$ and original endpoint $x \in \origin(u)$, $\delta(x, u)$ denotes the distance between $x$ and $u$ along the subdivision path.
    \begin{itemize}
    \item If $u \in \tilde{V}_i\setminus V$,  
        suppose that $u$ was created by subdividing the edge $e=(v,w) \in E$ into $k$ edges $e_1,\dots,e_k$, and let $u = u_j$ with $1 \le j < k$ be the $j$-th vertex along this path (with $u_0 = v$ and $u_k = w$).
        Then for each vertex in $\origin(u) = \{v,w\}$, we define
        \begin{equation*}
                \delta(u_j,v) := \sum_{l=1}^{j} \w_{\tilde{G}_i}(e_l) = \frac{j}{k} \w_G(e)\text{,}\; \delta(u_j,w) := \sum_{l=j+1}^{k} \w_{\tilde{G}_i}(e_l) = \frac{k - j}{k} \w_G(e).
        \end{equation*}       
    \item Otherwise, if $u \in V$, then $\delta(u,u) := 0$.
    \end{itemize}   
\end{restatable}
Using these definitions we define the auxiliary edge set $\tilde{F}_i$ of $\tilde{G}_i$, as
\begin{equation*}
    \tilde{F}_i = \{ (u,x) \mid u\in \tilde{V}_i \setminus V  \wedge  x\in \origin(u)\}
\end{equation*}
where the weight of each edge is $\w_{\tilde{F}_i}((u,x)) := \delta(u,x)$.
Before we continue with the construction, observe the following important properties of the subdivision graph $\tilde{G}_i$. 
\begin{observation}\label{obs:weight_tilde-pi}
    Let $\pi$ be any path from $u$ to $v$ in $G$. 
    There exists a path $\tilde{\pi}$ in $\tilde{G}_i$ between $u$ and $v$, such that $\w_G(\pi) =\w_{\tilde{G}_i}(\tilde{\pi})$.  
\end{observation}
\begin{proof}
    This follows immediately from the construction of $\tilde{G}_i$, where each edge of $G$ satisfying Eq.~\ref{equa:subdivisionCondition} is replaced by a sequence of new edges and intermediate vertices in $\tilde{G}_i$.
    Thus, the path $\tilde{\pi}$ can be obtained from $\pi$ by replacing each such edge with its corresponding sequence of new edges and vertices.
    Therefore, $\tilde{\pi}$ connects $u$ to $v$ in $\tilde{G}_i$ and the vertex set of $\tilde{\pi}$ consists of the vertices of $\pi$, together with any intermediate vertices introduced during the subdivision step.

    Each sequence of new edges in $\tilde{G}_i$, replacing an edge, has total weight equal to that of the original edge in $G$. Thus $\w_{\tilde{G}_i}(\tilde{\pi}) = \w_G(\pi)$ holds.
\end{proof}

\begin{claim}\label{claim:maxEdgeWeightSubG}
    For every pair of vertices $u,v \in V$ with $\dist_G(u,v) \le t (1+\epsilon)^{i + 1}$, every edge on the shortest path from $u$ to $v$ in $\tilde{G}_i$ has weight at most $\epsilon(1+\epsilon)^i$.
\end{claim}
\begin{proof}
    Let $\tilde{\pi}$ be the shortest path from $u$ to $v$ in $\tilde{G}_i$, which by Observation~\ref{obs:weight_tilde-pi} has weight equal to $\dist_G(u,v)$.
    Consider an edge $\tilde{e}$ on $\tilde{\pi}$. 
    If $\tilde{e} \notin E$, then by construction of $\tilde{G}_i$, $\tilde{e}$ is one of $k = \ceil{\frac{\w_G(e)}{\epsilon(1+\epsilon)^i}}$ edges in the path that replaces some edge $e \in E$, that satisfied Eq.~\ref{equa:subdivisionCondition}. Each of these $k$ edges has weight $\frac{\w_G(e)}{k} \le \epsilon(1+\epsilon)^i$.
    Otherwise, if $\tilde{e} \in E$ we know that $\tilde{e}$ doesn't satisfy Eq.~\ref{equa:subdivisionCondition}. This implies $\w_{\tilde{G}_i}(\tilde{e}) \le \epsilon(1+\epsilon)^{i+1}$ or $\w_{\tilde{G}_i}(\tilde{e}) > t(1+\epsilon)^{i+1}$. However, since $\dist_{\tilde{G}_i}(u,v) \le t(1+\epsilon)^{i+1}$, the weight of $\tilde{e}$ is bounded by $\w_{\tilde{G}_i}(\tilde{\pi}) \le t(1+\epsilon)^{i+1}$, and thus we know that $\tilde{e}$ has weight at most $\epsilon(1+\epsilon)^{i}$.
\end{proof}

\begin{claim}\label{claim:numberEdgesSubG}
    The graph $\tilde{G}_i$ contains at most $\bigO(m\frac{t}{\epsilon})$ edges and at most $\bigO(n+m\frac{t}{\epsilon})$ vertices.
\end{claim}
\begin{proof}
    An edge $e \in E$ is subdivided if its weight satisfies Eq.~\ref{equa:subdivisionCondition}, i.e. $\epsilon (1+\epsilon)^i < \w_G(e) \le t(1+\epsilon)^{i+1}$. In this case, it is replaced by $k = \ceil{\frac{\w_G(e)}{\epsilon(1+\epsilon)^i}}$ new edges. Since $\w_G(e) \le t(1+\epsilon)^{i+1}$,  $e$ is divided into at most $\ceil{\frac{t(1+\epsilon)^{i+1}}{\epsilon(1+\epsilon)^i}} =\ceil{\frac{t(1+\epsilon)}{\epsilon}} \le \frac{t}{\epsilon} (1+\epsilon) + 1 \le \frac{2t}{\epsilon} + 1$ edges in $\tilde{E}_i$. 
    Hence each edge in $E$ contributes to $\bigO(\frac{t}{\epsilon})$ edges, so the total number of edges in $\tilde{G}_i$ is bounded by $|\tilde{E}_i| = \bigO(m\frac{t}{\epsilon})$, and the number of vertices by $|\tilde{V}_i| = \bigO(n+m\frac{t}{\epsilon})$.
\end{proof}

\begin{lemma}\label{lem:numberSubdivisions}
    An edge $e \in E$ can be subdivided in at most $\log_{1+\epsilon}(\frac{t}{\epsilon}) + 1$ different graphs $\tilde{G}_i$.
\end{lemma}
\begin{proof}
    By definition, an edge $e \in E$ is subdivided at scale $i$ if its weight satisfies Eq.~\ref{equa:subdivisionCondition}, i.e. $\epsilon (1+\epsilon)^i < \w_G(e) \le t(1+\epsilon)^{i+1}$.
    For a fixed edge $e \in E$, we want to bound the number of integer values $i$ for which this inequality holds.
    Rewriting the bounds we obtain:
    \begin{align*}
        \epsilon (1+\epsilon)^i < \w_G(e) &\iff i < \log_{1+\epsilon}\left( \frac{\w_G(e)}{\epsilon} \right) \\
        \w_G(e) \le t(1+\epsilon)^{i+1} &\iff i \ge \log_{1+\epsilon}\left( \frac{\w_G(e)}{t} \right) - 1.
    \end{align*}
    Thus the number of integer values of $i$ satisfying Eq.~\ref{equa:subdivisionCondition} is less than $\log_{1+\epsilon}(\frac{\w_G(e)}{\epsilon}) - (\log_{1+\epsilon}(\frac{\w_G(e)}{t}) - 1)
        = \log_{1+\epsilon}(\frac{t}{\epsilon}) + 1$.
    Therefore, each edge is subdivided in at most $\log_{1+\epsilon}(\frac{t}{\epsilon}) + 1$ distinct distance scales.
\end{proof}

\paragraph*{Hopset Definition}
Let $G_i := (\tilde{V}_i, E_i)$ be an unweighted graph, defined over the vertex set $\tilde{V}_i$ of $\tilde{G}_i$, which contains an edge for every pair of vertices whose distance in $\tilde{G}_i$ lies in the $i$-th distance scale. Formally, the edge set is given by $E_i :=\{(u,v) \in \tilde{V}^2_i \mid (1 + \epsilon)^i \le \dist_{\tilde{G}_i}(u,v) < (1+\epsilon)^{i+1}\}$.
Let $M_i$ be an $(\alpha,\beta)$-emulator of $G_i$. Let $M_i'$ be the graph obtained by rescaling the weights of $M_i$ by the upper bound of the current distance scale. Specifically, for each edge $e \in \edgesop(M_i)$ its weight in $M_i'$ is given by $\w_{M_i'} (e) := (1+\epsilon)^{i+1}\w_{M_i}(e)$.  
Finally, we define the hopset $F$ as the union of the edge sets of all $M_i'$ and the auxiliary sets $\tilde{F}_i$ across distance scales:  
\begin{equation*}
    F := \bigcup\limits_{i=0}^{i_{\text{max}}} \edgesop(M_i') \cup \tilde{F}_i.
\end{equation*}

\paragraph*{Extended Graph} 
Since we consider arbitrary emulators $M_i$ of the unweighted graphs $G_i = (\tilde{V_i},E_i)$, we make no assumptions about their internal structure. In particular, $M_i$ may contain edges incident to vertices in $\tilde{V}_i \setminus V$, which were introduced in the subdivision step and do not appear in the original graph $G = (V,E)$.
Consequently, a path in $M_i$ from $u \in V$ to $v \in V$ may contain such edges, which would not exist in a hopset defined solely on the vertex set $V$. 
Hence, to ensure all paths in $M_i$ correspond to a path in the hopset graph, we need to extend the vertex set of $G$ by including all new vertices of each graph $\tilde{G}_i$.
Let $G' := (V',E)$ be the extended graph, with 
\begin{equation*}
   V' := \bigcup\limits_{i=0}^{i_{\text{max}}} \tilde{V}_i.
\end{equation*}
Thus $F$ is a $V$-constrained hopset of the extended graph $G'$, and the resulting hopset graph is $H = (V', E \cup F)$.
Note that $G'$ is obtained from $G$ by only adding vertices. All vertices in $V'\setminus V$ are isolated in $G'$ and only become connected through the edges in $F$ in the graph $H$. Thus, the distances and all paths between vertices in $V$ are identical in $G$ and $G'$. Throughout Section~\ref{section:hop-reduction}, we may refer to shortest paths in $G$ for simplicity, even though the hopset is defined on $G'$.

\subsection{Size Analysis}\label{section:size-analysis}
Next, we analyze the size of the hopset $F$.  
\begin{lemma}\label{lem:sizeHopset}
    Let $\mathcal{A}$ be an algorithm that, for any $n$-vertex graph, constructs an $(\alpha,\beta)$-emulator with $S_{\mathcal{A}}(n,\alpha,\beta)$ edges.
    Then the construction in Section~\ref{section:construction} yields a hopset $F \subseteq V'\times V'$ which consists of $\bigO(S_{\mathcal{A}}(n+m\frac{t}{\epsilon}, \alpha,\beta) \log_{1+\epsilon}(nW) + m\frac{t}{\epsilon}\log_{1+\epsilon}(\frac{t}{\epsilon}))$ edges.
\end{lemma}
\begin{proof}
    Using the algorithm $\mathcal{A}$, each $(\alpha,\beta)$-emulator $M_i$ is constructed on the graph $G_i$. Let $n_i$ denote the number of vertices in $G_i$. From Claim~\ref{claim:numberEdgesSubG} we know that $n_i \le n +\frac{2t}{\epsilon} m$. Thus each emulator $M_i$ contains $\bigO(S_{\mathcal{A}}(n+m\frac{t}{\epsilon},\alpha,\beta))$ edges.
    There are $i_{\text{max}} \le \log_{1+\epsilon}(nW)$ different distance scales, therefore we have at most $\log_{1+\epsilon}(nW)$ of such emulators.
    By Lemma~\ref{lem:numberSubdivisions} and Claim~\ref{claim:numberEdgesSubG}, it follows that over all edges and distance scales the number of newly introduced vertices in the subdivision step is $\bigO(m\frac{t}{\epsilon}\log_{1+\epsilon}(\frac{t}{\epsilon}))$, so the size of the auxiliary edge set $\tilde{F}$ is also $\bigO(m\frac{t}{\epsilon}\log_{1+\epsilon}(\frac{t}{\epsilon}))$.\footnote{The edges in $\tilde{F}$ get removed in the projection back to a hopset on $G$. See Section~\ref{section:projection}.} 
    Thus the total size of the hopset $F$ is in $\bigO(S_{\mathcal{A}}(n+m\frac{t}{\epsilon}, \alpha,\beta) \log_{1+\epsilon}(nW) + m\frac{t}{\epsilon}\log_{1+\epsilon}(\frac{t}{\epsilon}))$.
\end{proof}
Note that the size bound stated in Lemma~\ref{lem:sizeHopset} is conservative, as it assumes that every edge gets subdivided at every distance scale $i$ into the maximum possible number of subedges.

\subsection{Hop Reduction}\label{section:hop-reduction}
Let $G=(V,E)$ be the original graph, and let $F$ be the edge set as defined in Section~\ref{section:construction}.
In this section we show that $F$ is indeed a $V$-constrained $(\alpha(1+24\epsilon),\space \bigO(\alpha t^2)\ln(nW))$-hopset of $G'$, for $t = \max(\frac{1}{\epsilon},\frac{\beta}{\epsilon})$.
Let $\pi$ be a shortest path from $u$ to $v$ in $G$. 
We will show that there exists a vertex $w$ on $\pi$ and an alternative path $\pi'$ from $u$ to $w$ in $H = (V', E \cup F)$ with the following properties:
\begin{enumerate}[label=(P\arabic*)]
    \item \label{item:dist_u_w} The distance from $u$ to $w$ in $G$ is at least $\Delta  t$, with $\Delta := \frac{\dist_G(u,v)}{t^2}$.  
    \item \label{item:alt_path_hops} The alternative path $\pi'$ consists of at most $3t (\alpha + 1)$ edges of $F$ and includes at most one additional edge from either $E$ or $F$.
    \item \label{item:alt_path_weight}The weight of the alternative path $\pi'$ in $H$ is at most $\alpha (1 + 24 \epsilon)$ times the distance from $u$ to $w$ in $G$. 
\end{enumerate}
Thus, by choosing to take such an alternative path $\pi'$ from $u$ to $w$ instead of a subpath of $\pi$ we take at most $3t (\alpha + 1) + 1$ hops, while decreasing the remaining distance to $v$ by at least $\Delta t = \frac{\dist_G(u,v)}{t}$. However, this introduces an approximation error of $\alpha(1+24\epsilon)$.
\begin{figure}[H]
    \centering
    \begin{tikzpicture}[scale=1.2]
        \coordinate (u)  at (0,0);
        \coordinate (w1)  at (4,0);
        \coordinate (wi)  at (6.5,0);
        \coordinate (v)  at (8,0);

        \draw[thick] (u) -- (v);

        \draw[thick, color=blue, decorate, 
              decoration={random steps, segment length=6pt, amplitude=2pt}]
            (u) .. controls (1,2) and (3,2) .. (w1)
            node[pos=0.6, above=6pt, right=6pt, font=\small]
            {$\le 3t (\alpha + 1) + 1$ hops};
        \draw[thick, color=blue, decorate, 
              decoration={random steps, segment length=6pt, amplitude=2pt}]
            (w1) .. controls (4.75,1.25) and (5.75,1.25) .. (wi)
            node[pos=0.6, above=6pt, right=6pt, font=\small]
            {$\le 3t (\alpha + 1) + 1$ hops};
        \draw[thick, color=blue, decorate, 
              decoration={random steps, segment length=6pt, amplitude=2pt}]
            (wi) .. controls (6.75,0.6) and (7.75,0.6) .. (v)
            node[pos=0.6, above=6pt, right=6pt, font=\small]
            {$\le 3t (\alpha + 1) + 1$ hops};

        \fill (u)   circle (1.5pt) node[below=10pt,anchor=mid] {$u$};
        \fill (w1)   circle (1.5pt) node[below=10pt,anchor=mid] {$w_1$};
        \draw ($(w1) + (1.25,0)$) node[below=10pt,anchor=mid] {$\dots$};
        \fill (wi)   circle (1.5pt) node[below=10pt,anchor=mid] {$w_i$};
        \draw ($(wi) + (0.75,0)$) node[below=10pt,anchor=mid] {$\dots$};
        \fill (v)   circle (1.5pt) node[below=10pt,anchor=mid] {$v$};

        \draw [decorate, decoration={brace, mirror, amplitude=10pt}]
            ($(u) - (0,0.5)$) -- ($(w1) - (0,0.5)$)
            node[midway, below=8pt] {$\ge \Delta t = \frac{\dist_G(u,v)}{t}$};
    \end{tikzpicture}
    \caption{
        Instead of following the shortest path from $u$ to $v$ in $G$ (black), we take a sequence of alternative paths in $H$ (blue). Each alternative path consists of at most $3t (\alpha + 1) + 1$ hops and skips at least a fraction $\frac{1}{t}$ of the remaining distance to $v$. 
    }
    \label{fig:hop_reduction}
\end{figure}
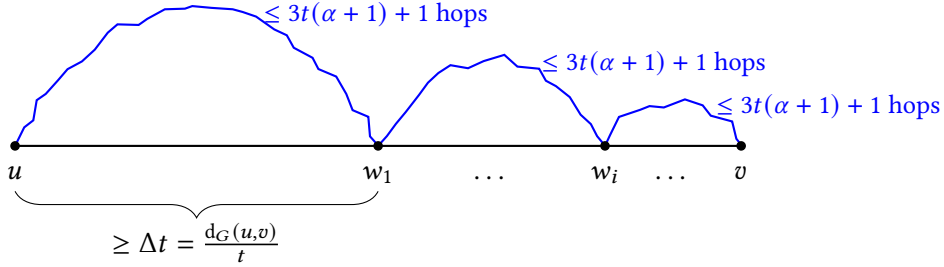
Through repeated application of this strategy with decreasing $\Delta$, and concatenating the alternative paths, we obtain an alternative path $\pi''$ in $H$ from $u$ to $v$ (see Figure~\ref{fig:hop_reduction}) that has at most $\ln(nW)\space (3t^2 (\alpha + 1) + t)$ hops. This follows from \ref{item:dist_u_w}, which guarantees that with each shortcut we skip a fixed fraction $\frac{1}{t}$ on the remaining path to $v$ and \ref{item:alt_path_hops}, which limits the number of hops in each shortcut. Property \ref{item:alt_path_weight} ensures that every alternative path has a multiplicative stretch of at most $\alpha (1 + 24\epsilon)$, and thus the whole path $\pi''$ also has the same multiplicative stretch. 
The idea of using such alternative paths to shortcut portions of the shortest path is similar to the approach used in \cite{Ber09, HKN16}.

\smallskip 
Before we proceed with the analysis as explained above, we prove two important structural properties of our construction from Section~\ref{section:construction}. 
For a fixed distance scale $i$, Claim~\ref{claim:up} ensures the existence of a vertex $u_p$ along the shortest path from $u$ to $v$ in the subdivided graph $\tilde{G}_i$, for which we can guarantee that there exists path from $u$ to $u_p$ in the unweighted graph $G_i$ with a bounded number of edges.
For any path in $G_i$, Claim~\ref{claim:weight_hop_emulatorpaths} provides a bound on the weight and number of edges of the corresponding shortest path in the emulator of this distance scale and hence in $H$. 

\begin{claim}\label{claim:up}
    Let $\Delta \ge 1$ and let $i$ be such that $(1+\epsilon)^i \le \Delta < (1+\epsilon)^{i+1}$.
    Consider any pair of vertices $u,v \in V$ with $\dist_G(u,v) < \infty$ and $\dist_G(u,v) \ge \Delta t$. 
    Then there exists a vertex $u_p \in \tilde{V}_i$ on the shortest path from $u$ to $v$ in $\tilde{G}_i$, for some $0 \le p \le t+t\epsilon$, such that $\dist_{\tilde{G}_i}(u,u_p) \le \Delta t$ and there is a path of $p$ edges from $u$ to $u_p$ in $G_i$. Hence $u_p$ satisfies $p(1+\epsilon)^i \le \dist_{\tilde{G}_i}(u,u_p) \le p(1+\epsilon)^{i+1}$.
    Moreover, there is no vertex $u_p'$ on the remaining subpath from $u_p$ to $v$, such that  $\dist_{\tilde{G}_i}(u_p,u_p') \ge (1+\epsilon)^i$ and $\dist_{\tilde{G}_i}(u, u_p') \leq \Delta t$. 
\end{claim}

\begin{proof}
    Let $\tilde{\pi}$ be the shortest path from $u$ to $v$ in $\tilde{G}_i$, which by Observation~\ref{obs:weight_tilde-pi} has weight equal to $\dist_G(u,v)$.
    We iteratively construct a sequence of vertices $u_0,u_1,\dots,u_p$ along $\tilde{\pi}$ starting at $u_0 := u$. For each integer $j \ge 0$ we define $u_{j+1}$ to be the first vertex on $\tilde{\pi}$ after $u_j$ that satisfies $\dist_{\tilde{G}_i}(u_j, u_{j+1}) \ge (1+\epsilon)^i$. If no such vertex exists or $\dist_{\tilde{G}_i}(u, u_{j+1}) > \Delta t$ we stop and set $u_p := u_j$. Otherwise, we continue the iteration.

    We now show by induction on $j$ that for every $0 \le j \le p$, there exists a path of $j$ edges from $u$ to $u_j$ in $G_i$. 
    \smallskip \\
    \textit{Base Case.}  $j = 0$. Since $u_0 = u$, the trivial path from $u$ to $u_0$ of length $0$ exists in $G_i$.
    \smallskip \\
    \textit{Induction Step.}
    Assume that for any $j < p$, there exists a path of $j$ edges from $u$ to $u_j$ in $G_i$.
    Let $x$ denote the last vertex on $\tilde{\pi}$ with $\dist_{\tilde{G}_i}(u_j,x) < (1+\epsilon)^i$, and let $u_{j+1}$ be its successor. By the choice of $x$, $u_{j+1}$ must satisfy $\dist_{\tilde{G}_i}(u_j,u_{j+1}) \ge (1+\epsilon)^i$. 
    Since $\dist_{\tilde{G}_i}(u,u_p) \le \Delta t \le t(1+\epsilon)^{i+1}$ Claim~\ref{claim:maxEdgeWeightSubG} guarantees that every edge $e$ on the subpath of $\tilde{\pi}$ from $u$ to $u_p$ has weight $\w_{\tilde{G}_i}(e) \le \epsilon(1+\epsilon)^i$. Thus, the distance between $u_j$ and $u_{j+1}$ can be upper bounded by $\dist_{\tilde{G}_i}(u_j,u_{j+1}) = \dist_{\tilde{G}_i}(u_j,x) +\w_{\tilde{G}_i}((x,u_{j+1})) < (1+\epsilon)^i +\epsilon(1+\epsilon)^i = (1+\epsilon)^{i+1}$.
    Combining these bounds, we have $(1+\epsilon)^i \le \dist_{\tilde{G}_i}(u_j,u_{j+1}) < (1+\epsilon)^{i+1}$. Therefore, by definition, $G_i$ contains the edge $(u_j,u_{j+1})$ and and hence there exists a path of $j+1$ edges from $u$ to $u_{j+1}$ in $G_i$.
    \smallskip \\
    Note that, since there is a path of $p$ edges from $u$ to $u_p$ in $G_i$ , we have $p(1+\epsilon)^i \le \dist_{\tilde{G}_i}(u,u_p) \le p(1+\epsilon)^{i+1}$.

    Since the constructed final vertex $u_p$ also satisfies $\dist_{\tilde{G}_i}(u, u_p) \le \Delta t<  t (1+\epsilon)^{i+1}$, we have $p(1+\epsilon)^i < t (1+\epsilon)^{i+1}$
    and it follows directly that $p < t (1+\epsilon) = t + t\epsilon$. 
\end{proof}

By the definition of $G_i$, the guarantees of the emulator $M_i$ and the definition of weights in $M_i'$ we get the following claim. 

\begin{claim}\label{claim:weight_hop_emulatorpaths}
    Consider a distance scale $i \ge 0$.
    If there exists a path of $l$ edges between two vertices $u$,$v$ in the graph $G_i$, then there exists a path $\pi_i$ in $H = (V', E\cup F)$, with $|\pi_i| \le \alpha l + \beta$ and $\w_H(\pi_i) \le (1+\epsilon)^{i+1} (\alpha l + \beta)$.
\end{claim}

\begin{proof}
    Let $\pi_i$ be the shortest path from $u$ to $v$ in the emulator $M_i$, which is by construction also in $F$.
    By definition, this path has weight $\w_{M_i}(\pi_i) = \dist_{M_i}(u,v) \le \alpha \cdot \dist_{G_i}(u,v) +\beta  \le \alpha l + \beta$.
    Since all edges have weight at least $1$, this yields 
    \begin{align*}
        |\pi_i| \le \alpha l + \beta.
    \end{align*}
    Recall that in $M_i'$ all edge weights are scaled by the factor $(1+\epsilon)^{i+1}$. Therefore the weight of the path $\pi_i$ in $H$ is $\w_H(\pi_i) =\w_{M_i'}(\pi_i) = (1+\epsilon)^{i+1} \w_{M_i}(\pi_i)$ and can be bounded by 
        \begin{equation*}
        \w_H(\pi_i) \le (1+\epsilon)^{i+1} (\alpha l + \beta).\qedhere
    \end{equation*}
\end{proof}

With Claim~\ref{claim:up} and Claim~\ref{claim:weight_hop_emulatorpaths}, we now show that we can construct a shortcut $\pi'$ with the desired properties. 

\begin{lemma}\label{lem1}
    Let $\Delta > 0$, for every pair of vertices $u,v \in V$ with $\dist_G(u,v) < \infty$ and $\dist_G(u,v) \ge \Delta t$, we can find a vertex $w$ on the shortest path from $u$ to $v$ in $G$, and a path $\pi'$ from $u$ to $w$ in $H = (V', E \cup F)$, such that $\dist_G(u,w) \ge \Delta t$, $|\pi'| \le 3t (\alpha + 1) + 1$ and $\w_H(\pi') \le \alpha (1 + 24\epsilon) \cdot \dist_G(u,w)$.
\end{lemma}

\begin{proof}
Let $\pi$ be the shortest path from $u$ to $v$ in $G$.
Let $w'$ be the closest vertex to $v$ on $\pi$ for which $\dist_G(u, w') < \Delta t$ still holds, and let $w$ be its successor on $\pi$.

\medskip 
In the following proof, we assume $\Delta \ge 1$, as otherwise the claim follows immediately: Given that all edge weights are at least $1$ and $\dist_G(u,w') < \Delta t < t$, the shortest path between $u$ and $w'$ in $G$ already consists of at most $t$ edges. We define the alternative path $\pi'$ to be the subpath of $\pi$ from $u$ to $w$. Then the guarantees for the number of hops and the stretch hold, since $|\pi'| = t+1 < 3t (\alpha + 1) + 1$ and $\w_H(\pi') = \dist_G(u,w) \le \alpha(1+24\epsilon) \dist_G(u,w)$.

\medskip 
For the remainder of the proof, consider $i$ such that $(1 + \epsilon)^i \le \Delta < (1+\epsilon)^{i+1}$, that is $\Delta$ lies in the $i$-th distance scale. 
Let $\tilde{\pi}$ be the shortest path from $u$ to $v$ in $\tilde{G}_i$, which by Observation~\ref{obs:weight_tilde-pi} has weight equal to $\dist_G(u,v)$ and consists of the vertices of $\pi$, together with any intermediate vertices introduced during the subdivision step. 
The strategy to construct the path $\pi'$ with the desired properties from $u$ to the vertex $w$ in this case is as follows:
First, we identify a vertex $u'$ on $\tilde{\pi}$ and a path $\pi_{\text{pre}}$ in $H$ from $u$ to $u'$ with at most $3t(\alpha+1)$ hops and a bounded stretch. Then we show that there exists a single edge from $u'$ to $w$ in either $E$ or in $F$, and by extending $\pi_{\text{pre}}$ with this edge, we obtain $\pi'$, whose stretch is at most $\alpha(1+24\epsilon)$.
Let $u_p$ be the vertex from Claim~\ref{claim:up}, so that $\dist_{\tilde{G}_i}(u,u_p) \le \Delta t$ and there exists a path of $p \le t+t\epsilon$ edges from $u$ to $u_p$ in $G_i$. Let $\pi_i$ be the path from $u$ to $u_p$ in $H$ guaranteed by Claim~\ref{claim:weight_hop_emulatorpaths}, with $\w_H(\pi_i) \le (1+\epsilon)^{i+1}(\alpha p + \beta)$ and $|\pi_i| \le \alpha p + \beta \le \alpha (t+t\epsilon) + \beta$.
Note that $\pi_i$ may be the trivial path of length $0$. We will address this case separately when relevant.

The definition of $u'$ and the path $\pi_{\text{pre}}$ depend on the position of $u_p$ relative to $w'$ along $\tilde{\pi}$ (see Figure~\ref{fig:cases}).
Intuitively, we choose $u'$ to be a vertex that is at least as far along $\tilde{\pi}$ as $w'$, so that we can reach $w$ with a single edge. 
Consider the following two cases:
\begin{itemize}
    \item Case~\ref{case:u_prime-is-u_p}: If $\dist_{\tilde{G}_i}(u, u_p) \ge \dist_{\tilde{G}_i}(u,w')$, we define $u' := u_p$.

    \item Case~\ref{case:u_prime-is-w_prime}: If $\dist_{\tilde{G}_i}(u, u_p) < \dist_{\tilde{G}_i}(u,w')$, we define $u' := w'$.
\end{itemize}

\begin{figure}[H]
    \centering
     \begin{minipage}[b]{0.45\textwidth}
        \centering
        \begin{tikzpicture}[scale=1, every node/.style={font=\small}]
            \coordinate (u)      at (0,0);
            \coordinate (up) at (3.5,0);    
            \coordinate (yprime) at (3,0); 
            \coordinate (deltat)     at (4,0);     
            \coordinate (y) at (4.5,0);   
            \coordinate (v)      at (6,0);
    
            \draw[dashed, thick] (u) -- (v);
    
            \fill (u) circle (1.5pt) node[below=4pt] {$u$};
            \fill (yprime) circle (1.5pt) node[below=2pt] {$w'$};

            \draw[thick] ($(up)+(0,-3pt)$) -- ($(up)+(0,3pt)$) node[below=8pt] {$u_p$};
            \draw[thick, red] ($(deltat)+(0,-5pt)$) -- ($(deltat)+(0,5pt)$) node[below=8pt] {$\Delta t$};
            \fill (y) circle (1.5pt) node[below=4pt] {$w$};
            \fill (v) circle (1.5pt) node[below=4pt] {$v$};            
        \end{tikzpicture}
        \caption*{(a) Case~\ref{case:u_prime-is-u_p}}
    \end{minipage}
    \hfill
    \begin{minipage}[b]{0.45\textwidth}
        \centering
        \begin{tikzpicture}[scale=1, every node/.style={font=\small}]
            \coordinate (u)      at (0,0);
            \coordinate (up) at (3,0);    
            \coordinate (yprime) at (3.5,0); 
            \coordinate (deltat)     at (4,0);     
            \coordinate (y) at (4.5,0);   
            \coordinate (v)      at (6,0);
    
            \draw[dashed, thick] (u) -- (v);
    
            \fill (u) circle (1.5pt) node[below=4pt] {$u$};
            \fill (yprime) circle (1.5pt) node[below=2pt] {$w'$};

            \draw[thick] ($(up)+(0,-3pt)$) -- ($(up)+(0,3pt)$) node[below=8pt] {$u_p$};
            \draw[thick, red] ($(deltat)+(0,-5pt)$) -- ($(deltat)+(0,5pt)$) node[below=8pt] {$\Delta t$};
            \fill (y) circle (1.5pt) node[below=4pt] {$w$};
            \fill (v) circle (1.5pt) node[below=4pt] {$v$};
        \end{tikzpicture}
        \caption*{(b) Case~\ref{case:u_prime-is-w_prime}}
    \end{minipage}
  
    \caption{The dashed line represents the shortest path from $u$ to $v$ in $\tilde{G}_i$. The original vertices $u,v,w', w \in V$ are a marked by a dot, while the vertex $u_p \in \tilde{V}_i  \subseteq V'$ is marked by a vertical line.}
    \label{fig:cases}
\end{figure}

\begin{case}\label{case:u_prime-is-u_p}
Suppose $\dist_{\tilde{G}_i}(u, u_p) \ge \dist_{\tilde{G}_i}(u,w')$.
In this case, the vertex $u_p$ already lies at or beyond $w'$ along the path $\tilde{\pi}$.
We therefore have $u' = u_p$ and define $\pi_{\text{pre}} := \pi_i$, as illustrated in Figure~\ref{fig:pi_pre_case1}.
Thus, we have $\w_H(\pi_{\text{pre}}) = \w_H(\pi_{i})  \le (1+\epsilon)^{i+1}(\alpha p + \beta)$ and $|\pi_{\text{pre}}| =  |\pi_i|\le \alpha (t+t\epsilon) + \beta = \alpha t + \alpha t\epsilon + \beta \le \alpha t + \alpha t \epsilon + t \epsilon \le 2t(\alpha + 1)$, since $\beta \le t\epsilon$ and $\epsilon \le 1$.

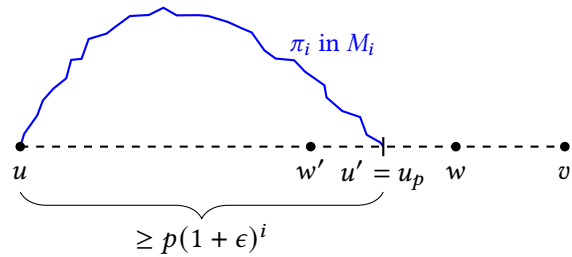
\begin{figure}[H]
    \centering
    \begin{tikzpicture}[scale=1.2]
        \coordinate (u)  at (0,0);
        \coordinate (v)  at (6,0);
        \coordinate (yp) at (3.2,0);
        \coordinate (up) at (4,0);   
        \coordinate (y) at (4.8,0);

        \draw[dashed, thick] (u) -- (v);

        \draw[thick, color=blue, decorate, 
              decoration={random steps, segment length=6pt, amplitude=2pt}]
            (u) .. controls (1,2) and (2,2) .. (up)
            node[pos=0.75, above=6pt, right=6pt, font=\small]
            {$\pi_{i}$ in $M_i$};

        \fill (u)   circle (1.5pt) node[below=10pt,anchor=mid] {$u$};
        \fill (yp)   circle (1.5pt) node[below=10pt,anchor=mid] {$w'$};
        \draw[black,thick] ($(up)+(0,-3pt)$) -- ($(up)+(0,3pt)$) node[below=14pt, anchor=mid] {$u' = u_p$};
        \fill (y)  circle (1.5pt) node[below=10pt,anchor=mid] {$w$};
        \fill (v)   circle (1.5pt) node[below=10pt,anchor=mid] {$v$};

        \draw [decorate, decoration={brace, mirror, amplitude=10pt}]
            ($(u) - (0,0.5)$) -- ($(up) - (0,0.5)$)
            node[midway, below=8pt] {$\ge p(1+\epsilon)^i$};
    \end{tikzpicture}
    \caption{
        The dashed line represents the shortest path from $u$ to $v$ in $\tilde{G}_i$.
        The vertices $u$, $u'$, and $v$ (all in $V$) are marked by dots.
        The solid blue line is the alternative path $\pi_{\text{pre}}$ from $u$ to $u'$ in $H$, which in this case only consists of $\pi_i$.
    }
    \label{fig:pi_pre_case1}
\end{figure}

By comparing the weight of the path $\pi_{\text{pre}}$ in $H$ to the exact distance from $u$ to $u'$ in $\tilde{G}_i$, we can bound the multiplicative stretch $\lambda_1$ which the path $\pi_{\text{pre}}$ introduces: 
\begin{align*}
    \lambda_1 :=\frac{\w_H(\pi_{\text{pre}})}{\dist_{\tilde{G}_i}(u,u')} &\le \frac{(1+\epsilon)^{i+1} (\alpha p +\beta)}{\dist_{\tilde{G}_i}(u,u')}.
\end{align*}
The vertex $u'$ satisfies $\dist_{\tilde{G}_i}(u,u') \ge p (1+\epsilon)^i$ by construction of $u_p$. Substituting this yields
\begin{align*}
    \lambda_1 = \frac{\w_H(\pi_{\text{pre}})}{\dist_{\tilde{G}_i}(u,u')}  &\le \frac{(1 + \epsilon)^{i+1} (\alpha p + \beta)}{p(1+\epsilon)^i} =  \frac{(1+\epsilon)(\alpha p +\beta)}{p} = (1+\epsilon) \left(\alpha + \frac{\beta}{p}\right).
\end{align*}
\end{case}

\begin{case}\label{case:u_prime-is-w_prime}
Suppose that  $\dist_{\tilde{G}_i}(u, u_p) < \dist_{\tilde{G}_i}(u,w')$.
In this situation, $u_p$ lies strictly before $w'$ along $\tilde{\pi}$. We have defined $u' = w'$ and to construct $\pi_{\text{pre}}$, we extend the path $\pi_i$ with an additional shortcut path from $u_p$ to $w'$, as shown in Figure~\ref{fig:pi_pre_case2}.
Since $\dist_{\tilde{G}_i}(u,w') < \Delta t$, the construction of $u_p$ in Claim~\ref{claim:up} implies that the remaining distance from $u_p$ to $w'$ is $\dist_{\tilde{G}_i}(u_p,w') < (1+\epsilon)^i$.

\begin{subcase}\label{subcase:trivial}
    We first consider the trivial case where $u_p = u$, thus $\pi_i$ is the empty path.
    Since $\dist_{\tilde{G}_i}(u,u') = \dist_{\tilde{G}_i}(u_p,w') < (1+\epsilon)^i$, we know that there exists a smaller distance scale $j < i$, such that $\dist_G(u,u')$ falls into this scale, i.e. $(1+\epsilon)^j \le \dist_G(u,u') < (1+\epsilon)^{j+1}$.
    Hence, by definition, we have an edge from $u$ to $u'$ in the graph $G_j$.
    Let $\pi_\text{pre}$ be defined as the path from $u$ to $u'$ in $H$ with $\w_H(\pi_\text{pre}) \le (1+\epsilon)^{j+1} (\alpha +\beta)$ and $|\pi_\text{pre}| = \alpha+\beta $, as provided by Claim~\ref{claim:weight_hop_emulatorpaths}.
    To bound the stretch of this shortcut, we compare the weight of $\pi_{\text{pre}}$ to the exact distance between $u$ and $u'$ in $G$. Since $(1+\epsilon)^j \le \dist_G(u,u')$, this yields a stretch $\lambda_{2.1}$ of 
    \begin{equation*}
        \lambda_{2.1} := \frac{\w_H(\pi_{\text{pre}})}{\dist_{\tilde{G}_i}(u,u')} \le \frac{(1+\epsilon)^{j+1}(\alpha +\beta)}{(1+\epsilon)^j}  = (1+\epsilon)(\alpha + \beta).
    \end{equation*}
\end{subcase}
\begin{subcase}\label{subcase:nontrivial}
    Next, we consider the case where $u \ne u_p$.
    Let $x$ be the vertex in $\origin(u_p)$ that succeeds $u_p$ on $\tilde{\pi}$, i.e., $x$ is the first vertex in $V$ on the subpath $\tilde{\pi}[u_p,v]$.

    We construct the path $\pi_\text{pre}$ by combining the path $\pi_i$ from $u$ to $u_p$ with a path $\pi_j$ from $x$ to $u'$. 
    If $u_p = x$, we can directly concatenate $\pi_i$ and $\pi_j$. Otherwise if $u_p \neq x$, by definition, there exists an edge $(u_p,x) \in F$ with weight $\dist_{\tilde{G}_i}(u_p, x)$ and we extend $\pi_i$ by $(u_p,x)$ before concatenating it with $\pi_j$. 
    
    The definition of $\pi_j$ is based on the following cases:
    \begin{itemize}
        \item If $x = u'$: Let $\pi_j$ be the empty path.
        \item If $x \ne u'$: We know that the remaining distance from $x$ to $u'$ is $\dist_{\tilde{G}_i}(x,u') < (1+\epsilon)^i$, since $\dist_{\tilde{G}_i}(u,x) \ge \dist_{\tilde{G}_i}(u,u_p)$ and $\dist_{\tilde{G}_i}(u, u') -\dist_{\tilde{G}_i}(u,u_p) < (1+\epsilon)^i$.
                    Hence $\dist_{\tilde{G}_i}(x, u')$ falls into some distance scale $j$ with $j < i$ and since $x, u' \in V$ there is, by definition, a single edge from $x$ to $u'$ in $G_j$. Let $\pi_j$ denote the path in $H$ with $\w_H(\pi_j) \le (1+\epsilon)^{j+1} (\alpha +\beta)$ from $x$ to $u'$, guaranteed by Claim~\ref{claim:weight_hop_emulatorpaths}.
    \end{itemize}
    
    Thus the number of edges in $\pi_\text{pre}$ is $\pi_\text{pre} \le |\pi_i| + |\pi_j| + 1$. Its weight is $\w_{H}(\pi_{\text{pre}}) \le  \w_{H}(\pi_{i}) + \w_{H}(\pi_{j}) +  \dist_{\tilde{G}_i}(u_p, x)$.

    \begin{figure}[h]
        \centering
        \begin{minipage}{0.6\textwidth}
            \centering
            \begin{tikzpicture}[scale=1]
                \coordinate (u)  at (0,0);
                \coordinate (v)  at (6,0);
                \coordinate (up) at (3,0);
                \coordinate (up2) at (4,0);
    
                \draw[dashed, thick] (u) -- (v);
    
                \draw[thick, blue, decorate,
                      decoration={random steps, segment length=6pt, amplitude=2pt}]
                    (u) .. controls (1,2) and (2,2) .. (up)
                    node[pos=0.75, above=6pt, right=6pt, font=\small] {$\pi_i \space \text{ in } \space M_i $};
                \draw[thick, blue, decorate,
                      decoration={random steps, segment length=4pt, amplitude=2pt}]
                    (up) .. controls (3.25,0.5) and (3.75,0.5) .. (up2)
                    node[pos=0.75, above=6pt, right=6pt, font=\small] {$\pi_j \space \text{ in } \space M_j$};

                \fill (u) circle (1.5pt) node[below=4pt] {$u$};
                \fill (up) circle (1.5pt) node[below=4pt] {$u_p = x$};
                \fill (up2) circle (1.5pt) node[below=2pt] {$u'$};
                \fill (v) circle (1.5pt) node[below=4pt] {$v$};

                \draw [decorate,decoration={brace,mirror,amplitude=10pt}]
                    ($(u)-(0,0.6)$) -- ($(up)-(0,0.6)$)
                    node[midway,below=8pt] {$\ge p(1+\epsilon)^i$};
                \draw [decorate,decoration={brace,mirror,amplitude=10pt}]
                    ($(up)-(0,0.6)$) -- ($(up2)-(0,0.6)$)
                    node[midway,below=8pt] {$< (1+\epsilon)^i$};
            \end{tikzpicture}
            \caption*{(a) $u_p \in V$}
        \end{minipage}
    
        \vspace{0.4em}
    
        \begin{minipage}{0.45\textwidth}
            \centering
            \begin{tikzpicture}[scale=1, every node/.style={font=\small}]
                \coordinate (u) at (0,0);
                \coordinate (up) at (3.5,0);
                \coordinate (upp) at (4.5,0);
                \coordinate (v) at (6,0);
    
                \draw[dashed, thick] (u) -- (v);
                
                \draw[thick, blue, decorate,
                      decoration={random steps, segment length=6pt, amplitude=2pt}]
                    (u) .. controls (1,2) and (2,2) .. (up)
                    node[pos=0.7, above=6pt, right=4pt, font=\small]  {$\pi_i \space \text{ in } \space M_i $};

                \draw[thick, blue] (up) -- (upp) node[pos=0.75, above=2pt] {$(u_p,x) \in F$};

                \fill (u) circle (1.5pt) node[below=4pt] {$u$};
                \draw[thick] ($(up)+(0,-3pt)$) -- ($(up)+(0,3pt)$) node[below=8pt] {$u_p$};
                \fill (upp) circle (1.5pt) node[below=2pt] {$x = u'$};
                \fill (v) circle (1.5pt) node[below=4pt] {$v$};

                \draw [decorate,decoration={brace,mirror,amplitude=10pt}]
                    ($(u)-(0,0.6)$) -- ($(up)-(0,0.6)$)
                    node[midway,below=8pt] {$\ge p(1+\epsilon)^i$};
                \draw [decorate,decoration={brace,mirror,amplitude=10pt}]
                    ($(up)-(0,0.6)$) -- ($(upp)-(0,0.6)$)
                    node[midway,below=8pt] {$< (1+\epsilon)^i$};
            \end{tikzpicture}
            \caption*{(b) $u_p \in \tilde{V}_i\setminus V$ and $x = u'$}
        \end{minipage}
        \hfill
        \begin{minipage}{0.45\textwidth}
            \centering
            \begin{tikzpicture}[scale=1, every node/.style={font=\small}]
                \coordinate (u) at (0,0);
                \coordinate (up) at (3.5,0);
                \coordinate (upp) at (4,0);
                \coordinate (uprime) at (4.8,0);
                \coordinate (v) at (6,0);
    
                \draw[dashed, thick] (u) -- (v);
    
                \draw[thick, blue, decorate,
                      decoration={random steps, segment length=6pt, amplitude=2pt}]
                    (u) .. controls (1,2) and (2,2) .. (up)
                    node[pos=0.7, above=6pt, right=4pt, font=\small] {$\pi_i \space \text{ in } \space M_i $};

                \draw[thick, blue] (up) -- (upp);
                \node[above=15pt, xshift=5pt, font=\small, inner sep=0pt, blue] (label) at ($(up)!0.5!(upp)$) {$(u_p,x) \in F$};
                \draw[->, >=stealth, ultra thin, gray] (label.south) -- ($(up)!0.5!(upp)+(0,2pt)$);

                \draw[thick, blue, decorate,
                      decoration={random steps, segment length=3pt, amplitude=2pt}]
                    (upp) .. controls (4.1,0.5) and (4.4,0.5) .. (uprime)
                                  node[pos=0.75, above=6pt, right=6pt, font=\small] {$\pi_j \space \text{ in } \space M_j $};
                                  
                \fill (u) circle (1.5pt) node[below=4pt] {$u$};
                \draw[thick] ($(up)+(0,-3pt)$) -- ($(up)+(0,3pt)$) node[below=8pt] {$u_p$};
                \fill (upp) circle (1.5pt) node[below=4pt] {$x$};
                \fill (uprime) circle (1.5pt) node[below=2pt] {$u'$};
                \fill (v) circle (1.5pt) node[below=4pt] {$v$};

                \draw [decorate,decoration={brace,mirror,amplitude=10pt}]
                    ($(u)-(0,0.6)$) -- ($(up)-(0,0.6)$)
                    node[midway,below=8pt] {$\ge p(1+\epsilon)^i$};
                \draw [decorate,decoration={brace,mirror,amplitude=10pt}]
                    ($(up)-(0,0.6)$) -- ($(uprime)-(0,0.6)$)
                    node[midway,below=8pt] {$< (1+\epsilon)^i$};
            \end{tikzpicture}
            \caption*{(c) $u_p \in \tilde{V}_i\setminus V$ and $x \ne u'$}
        \end{minipage}
    
        \caption{
            The dashed line represents the shortest path from $u$ to $v$ in $\tilde{G}_i$. Vertices in $V$ are marked by a dot, while vertices in $\tilde{V}_i \setminus V$ are marked by a vertical line. The solid blue line is the alternative path $\pi_\text{pre}$ from $u$ to $u'$ in $H$, consisting of subpaths $\pi_i$ and $\pi_j$. Furthermore, if $u_p \in \tilde{V}_i\setminus V$ the path $\pi_\text{pre}$ is formed by extending the path $\pi_i$ with the edge $(u_p, x) \in F$, and $\pi_j$. This can be seen in (b) and (c) (the subpath $\pi_j$ in (b) is empty).
        }
        \label{fig:pi_pre_case2}
    \end{figure}
    
    Using $|\pi_i| \le  \alpha (t+t\epsilon) + \beta$ and $|\pi_j| \le \alpha + \beta$, the number of edges can be bounded as follows
    \begin{align*}
        |\pi_{\text{pre}}| \le  |\pi_i|  +|\pi_j| + 1 &\le \alpha (t+t\epsilon) + \beta + \alpha +\beta + 1 = \alpha t + \alpha t\epsilon+\alpha +2\beta +1 \\
        &\le \alpha t + \alpha t\epsilon + \alpha t + 2t\epsilon  + t \le 3t(\alpha+1).
    \end{align*}
    
    To analyze the stretch of this shortcut, we compare the weight of $\pi_\text{pre}$ to the exact distance from $u$ to $u'$ in $\tilde{G}_i$.
    We know that, if present, the edge $(u_p, x)$ has weight $\dist_{\tilde{G}_i}(u_p, x)$, which does not contribute to the stretch of the shortcut path. Otherwise $u_p = x$ and therefore $\dist_{\tilde{G}_i}(u_p,x) = 0$. Hence, to obtain the stretch $\lambda_{2.2}$, it suffices to bound the ratio
    \begin{align*}
        \lambda_2 := \frac{\w_H(\pi_{\text{pre}}) - \dist_{\tilde{G}_i}(u_p,x)}{\dist_{\tilde{G}_i}(u,u') -  \dist_{\tilde{G}_i}(u_p,x)}.
    \end{align*}
    Since $\dist_{\tilde{G}_i}(u,u') -  \dist_{\tilde{G}_i}(u_p,x)> \dist_{\tilde{G}_i}(u,u_p)$ and
    $\dist_{\tilde{G}_i}(u,u_p) > p(1+\epsilon)^i$, this yields:
    \begin{align*}
            \lambda_{2.2} =  \frac{\w_H(\pi_{\text{pre}}) - \dist_{\tilde{G}_i}(u_p,x)}{\dist_{\tilde{G}_i}(u,u') -  \dist_{\tilde{G}_i}(u_p,x)} &\le \frac{\w_H(\pi_i) + \w_H(\pi_j) + \dist_{\tilde{G}_i}(u_p,x) - \dist_{\tilde{G}_i}(u_p,x)}{\dist_{\tilde{G}_i}(u,u') - \dist_{\tilde{G}_i}(u_p,x)} \\
            &\le \frac{\w_H(\pi_i) + \w_H(\pi_j) }{p(1+\epsilon)^i}.
    \end{align*}
    Further since $\w_H(\pi_i) \le (1+\epsilon)^{i+1} (\alpha p+ \beta)$ and $\w_H(\pi_j) \le (1+\epsilon)^{j+1} (\alpha + \beta) < (1+\epsilon)^{i+1} (\alpha + \beta)$, we have 
    \begin{align*}
        \w_{H}(\pi_{i}) + \w_{H}(\pi_{j}) &\le (1+\epsilon)^{i+1}(\alpha p +\beta) + (1+\epsilon)^{i+1}(\alpha+\beta) 
        = (1+\epsilon)^{i+1}  (\alpha p + \alpha + 2 \beta).
    \end{align*}
    This allows us to bound $\lambda_{2.2}$ by
    \begin{align*}
        \lambda_{2.2} &\le \frac{\w_H(\pi_i) + \w_H(\pi_j) }{p(1+\epsilon)^i} \le \frac{(1+\epsilon)^{i+1}  (\alpha p + \alpha + 2 \beta) }{p(1+\epsilon)^i} \\
        &= \frac{(1+\epsilon)  (\alpha p + \alpha + 2 \beta) }{p} = (1+\epsilon) \left(\alpha + \frac{\alpha}{p}+2\frac{\beta}{p}\right).
    \end{align*}
\end{subcase}
\end{case}

\medskip \noindent
Having defined $u'$ and $\pi_\text{pre}$, it remains to show that we can reach the vertex $w$ on $\pi$ with a single edge from $u'$ to $w$.
\begin{itemize}
    \item If $u'\ne w'$: 
    This only occurs in Case~\ref{case:u_prime-is-u_p} where $u' = u_p$ and we have $\dist_{\tilde{G}_i}(u,w') < \dist_{\tilde{G}_i}(u,u_p)$. Since $\dist_{\tilde{G}_i}(u,u_p) \le \Delta t \le \dist_{\tilde{G}_i}(u,w)$, the vertex $u_p$ must have been introduced during the subdivision of the edge $(w',w) \in E$. Therefore, in our construction, we have inserted the edge $(u', w)$ into $\tilde{F}_i \subseteq F$, with weight $\w_F((u', w)) = \delta(u', w)  =\dist_{\tilde{G}_i}(u',w)$. 
    
    \item If $u' = w'$:
    This can occur in either Case~\ref{case:u_prime-is-u_p} or Case~\ref{case:u_prime-is-w_prime}.
    We simply take the edge $(u', w) = (w',w) \in E$.  
\end{itemize}

Let $\pi'$ be the path in $H$ obtained by extending $\pi_{\text{pre}}$ with the edge $(u',w) \in E$ or $(u',w) \in F$, as illustrated in Figure~\ref{fig:altpath}.
\begin{figure}[H]
    \centering
    \begin{minipage}[b]{0.45\textwidth}
        \centering
        \begin{tikzpicture}[scale=1]
            \coordinate (u) at (0,0);
            \coordinate (v) at (5,0);
            \coordinate (yp) at (3,0);
            \coordinate (up) at (3.4,0);  
            \coordinate (y) at (4,0); 
            
            \draw[dashed, thick] (u) -- (y);
            \draw[dashed, thick] (y) -- (v);
            
            \draw[thick,red] 
                (y) -- (up)
                node[pos=0.1, above=1pt, font=\small] {$(u',w)\in F$};
            
            \draw[thick, decorate, decoration={random steps, segment length=6pt, amplitude=2pt},color=blue]
                (u) .. controls (0.5,2) and (1.5,2) .. (up)
                node[pos=3/4, above=6pt, right=6pt, font=\small] {$\pi_{\text{pre}} \text{ in } H$};
            
            \fill (u) circle (1.5pt) node[below=10pt, anchor=mid] {$u$};
            \fill (v) circle (1.5pt) node[below=10pt, anchor=mid] {$v$};
            \fill (yp) circle (1.5pt) node[below=10pt, anchor=mid] {$w'$};
            \draw[black,thick] ($(up)+(0,-3pt)$) -- ($(up)+(0,3pt)$) node[below=14pt, anchor=mid] {$u'$};
            \fill (y) circle (1.5pt) node[below=10pt, anchor=mid] {$w$};
            
            \draw [decorate,decoration={brace,mirror,amplitude=10pt}]
                ($(u) - (0,0.4)$) -- ($(y) - (0,0.4)$) node[midway,below=8pt] {$\ge \Delta t$};
        \end{tikzpicture}
        \caption*{(a) $u' \in \tilde{V}_i\setminus V$}
    \end{minipage}
    \hfill
    \begin{minipage}[b]{0.45\textwidth}
        \centering
        \begin{tikzpicture}[scale=1]
            \coordinate (u) at (0,0);
            \coordinate (v) at (5,0);
            \coordinate (up) at (3,0);  
            \coordinate (y) at (4,0); 
            
            \draw[dashed, thick] (u) -- (up);
            \draw[dashed, thick] (y) -- (v);
            
            \draw[thick,red] 
                (up) -- (y)
                node[pos=0.8, above=1pt, font=\small] {$(u',w)\in E$};
    
            \draw[thick, decorate, decoration={random steps, segment length=6pt, amplitude=2pt},color=blue]
                (u) .. controls (0.5,2) and (1.5,2) .. (up)
                node[pos=3/4, above=6pt, right=6pt, font=\small] {$\pi_{\text{pre}} \text{ in } H$};
            
            \fill (u) circle (1.5pt) node[below=10pt,anchor=mid]{$u$};
            \fill (v) circle (1.5pt) node[below=10pt,anchor=mid]{$v$};
            \fill (up) circle (1.5pt) node[below=10pt,anchor=mid] {$w' = u'$};
            \fill (y) circle (1.5pt) node[below=10pt,anchor=mid] {$w$};
            
            \draw [decorate,decoration={brace,mirror,amplitude=10pt}]
                ($(u) - (0,0.4)$) -- ($(y) - (0,0.4)$) node[midway,below=8pt] {$\ge \Delta t$};
        \end{tikzpicture}
        \caption*{(b) $u' = w' \in V$}
    \end{minipage}
    
    \caption{The dashed line represents the shortest path from $u$ to $v$ in $\tilde{G_i}$, while the solid line is the alternative path $\pi'$ from $u$ to $w$ in $H$. The original vertices in $V$ are marked by dots, whereas nodes introduced during the subdivision step are represented by vertical lines. The blue line shows the path $\pi_{\text{pre}}$ and the red line is the edge $(u',w)$ in $E$ or $F$ depending on the case.}
    \label{fig:altpath}
\end{figure}
We have previously shown that in all cases $|\pi_{\text{pre}}| \le 3t(\alpha+1)$. Thus the path $|\pi'|$ consists of at most $3t(\alpha+1) + 1$ edges as claimed.
Depending on the case, we showed that $\w_H(\pi_{\text{pre}}) \le \lambda \cdot \dist_{\tilde{G}_i}(u,u')$ for $\lambda \in \{\lambda_1,\lambda_{2.1},\lambda_{2.2}\}$.
Since the last edge $(u',w)$ of the path $\pi$ is exact, it introduces no additional error, and the weight of $\pi'$ can be expressed as
\begin{align*}
    \w_H(\pi') &= \w_H(\pi_{\text{pre}}) + \w_H((u', w)) \le \lambda \cdot \dist_{\tilde{G}_i}(u,u') + \dist_{\tilde{G}_i}(u',w) \\
    &= \lambda \cdot  \dist_{\tilde{G}_i}(u,u') + (\dist_{\tilde{G}_i}(u,w) - \dist_{\tilde{G}_i}(u,u')) = \dist_{\tilde{G}_i}(u,w) + (\lambda - 1) \cdot \dist_{\tilde{G}_i}(u,u').
\end{align*}
This allows us to bound the stretch of $\pi'$ by comparing its weight to the exact distance from $u$ to $w$ in $\tilde{G}_i$ (which coincides with the distance in $G$). Since $\dist_{\tilde{G}_i}(u,w) \ge \Delta t \ge t (1+\epsilon)^i$, we obtain the following:
\begin{equation*}
     \frac{\w_H(\pi')}{\dist_{\tilde{G}_i}(u,w)} \le \frac{\dist_{\tilde{G}_i}(u,w) + (\lambda - 1) \cdot \dist_{\tilde{G}_i}(u,u')}{\dist_{\tilde{G}_i}(u,w)} \le 1 + \left( \lambda - 1 \right) \frac{ \dist_{\tilde{G}_i}(u,u')}{\dist_{\tilde{G}_i}(u,w)} \le 1 + \left( \lambda - 1 \right) \frac{ \dist_{\tilde{G}_i}(u,u')}{t (1+\epsilon)^i}
\end{equation*}

Now, it remains to show that in all cases (Case~\ref{case:u_prime-is-u_p}, \ref{subcase:trivial} and \ref{subcase:nontrivial}) for the corresponding $\lambda$ and an appropriate upper bound for $\dist_{\tilde{G}_i}(u,u')$, we can bound the stretch of $\pi'$ with $\alpha(1 + 24 \epsilon)$.

First, we consider the simple Case~\ref{subcase:trivial}. Here we derived the stretch $\lambda_{2.1} = (1+\epsilon)(\alpha+\beta)$ for the path $\pi_{\text{pre}}$. Because in this case, we have $\dist_{\tilde{G}_i}(u,u') < (1+\epsilon)^{i}$, the stretch of $\pi'$ can be bounded as follows
\begin{align*}
    \frac{\w_H(\pi')}{\dist_{\tilde{G}_i}(u,w)} &\le  1 + \left( (\lambda_{2.1} - 1 \right) \frac{(1+\epsilon)^i}{t(1+\epsilon)^i} \le  1 + \left( (1+\epsilon) (\alpha+\beta) - 1 \right) \frac{(1+\epsilon)^i}{t(1+\epsilon)^i} \\
    &= 1 + \left( (1+\epsilon)\alpha +  (1+\epsilon)\beta - 1 \right) \frac{1}{t} = 1 + \left( (1+\epsilon)\alpha - 1 \right) \frac{1}{t} +  (1+\epsilon)\beta \frac{1}{t} \\
\end{align*}
As $t = \max(\frac{1}{\epsilon},\frac{\beta}{\epsilon})$, we have $\frac{1}{t}  \le \epsilon \le 1$ and $\frac{\beta}{t} \le \epsilon$. This yields
\begin{align*}
    \frac{\w_H(\pi')}{\dist_{\tilde{G}_i}(u,w)} &\le 1 + (1+\epsilon)\alpha - 1 + (1+\epsilon) \epsilon = \alpha + \alpha \epsilon + \epsilon +\epsilon^2 \le \alpha (1+\epsilon + \epsilon +\epsilon^2) \\&\le \alpha(1+3\epsilon).
\end{align*}

Next, we examine the remaining cases. In Case~\ref{case:u_prime-is-u_p} we know that $\dist_{\tilde{G}_i}(u,u') \le p (1+\epsilon)^{i+1}$ and derived the stretch $\lambda_1 = (1+\epsilon)(\alpha+\frac{\beta}{p})$.
In Case~\ref{subcase:nontrivial} we have $\dist_{\tilde{G}_i}(u,u') \le p (1+\epsilon)^{i+1} + (1+\epsilon)^i$ and the stretch $\lambda_{2.2} = (1+\epsilon)(\alpha + \frac{\alpha}{p} + 2 \frac{\beta}{p})$.
Since $\lambda_{2.2} > \lambda_1$ and $\dist_{\tilde{G}_i}(u,u')$ in Case~\ref{subcase:nontrivial} is larger than in Case~\ref{case:u_prime-is-u_p}, we can bound the stretch of $\pi'$ for these cases as follows: 
\begin{align*}
    \frac{\w_H(\pi')}{\dist_{\tilde{G}_i}(u,w)} &\le 1 + \left( \lambda_{1} - 1 \right) \frac{ p(1+\epsilon)^{i+1}}{t (1+\epsilon)^i} \le 1+(\lambda_{2.2} - 1) \frac{p(1+\epsilon)^{i+1}+(1+\epsilon)^i}{t(1+\epsilon)^i}\\
    &= 1+(\lambda_{2.2} - 1) \left(\frac{p}{t}(1+\epsilon) + \frac{1}{t}\right)= 1+(\lambda_{2.2}- 1) \frac{p}{t}(1+\epsilon) + (\lambda_{2.2} - 1) \frac{1}{t}.
\end{align*}
The term $(\lambda_2 - 1) \frac{p}{t}(1+\epsilon)$, can be expressed as
\begin{align*}
    (\lambda_2 - 1) \frac{p}{t}(1+\epsilon)&= \left((1+\epsilon)(\alpha+\frac{\alpha}{p}+2\frac{\beta}{p} )-1\right)\frac{p}{t}(1+\epsilon) \\
    &= \left(\alpha(1+\epsilon) + \frac{\alpha}{p}(1+\epsilon)+2\frac{\beta}{p}(1+\epsilon)-1\right)\frac{p}{t}(1+\epsilon) \\
    &=  (\alpha(1+\epsilon) - 1) \frac{p}{t} (1+\epsilon) + \frac{\alpha}{p}\frac{p}{t}(1+\epsilon)^2+2\frac{\beta}{p}\frac{p}{t}(1+\epsilon)^2.
\end{align*}
Given that $\frac{p}{t} \le \frac{t+t\epsilon}{t} = 1+\epsilon$, $\frac{\alpha}{p}\frac{p}{t} \le \alpha \epsilon$ and $\frac{\beta}{p}\frac{p}{t} \le \epsilon$, we get
\begin{align*}
      (\lambda_2 - 1) \frac{p}{t}(1+\epsilon) &\le  (\alpha(1+\epsilon) - 1) (1+\epsilon)^2 +\alpha \epsilon (1+\epsilon)^2+2\epsilon(1+\epsilon)^2\\
    &= \alpha(1+\epsilon)^3 - (1+\epsilon)^2 +\alpha \epsilon (1+\epsilon)^2+2\epsilon(1+\epsilon)^2\\
    &= \alpha(1+\epsilon)^3 - 1 - 2\epsilon -\epsilon^2 +\alpha \epsilon (1+\epsilon)^2+2\epsilon+4\epsilon^2+2\epsilon^3\\
    &= - 1+\alpha(1+\epsilon)^3  +\alpha \epsilon (1+\epsilon)^2+3\epsilon^2+2\epsilon^3.
\end{align*}
Since $\frac{1}{p} \le 1$, $\frac{\beta}{t} \le \epsilon$ and $\frac{1}{t} \le \epsilon$, the term $(\lambda_2 - 1) \frac{1}{t}$ is bounded by
\begin{align*}
    (\lambda_2 - 1) \frac{1}{t} &\le \lambda_2 \frac{1}{t} = (1+\epsilon)\left(\alpha+\frac{\alpha}{p}+2\frac{\beta}{p} \right) \frac{1}{t} \\
    &\le  (1+\epsilon)(2\alpha + 2\beta) \frac{1}{t}= (1+\epsilon) \left(2\alpha\frac{1}{t} + 2\frac{\beta}{t}\right) \\
    &\le (1+\epsilon) (2\alpha\epsilon + 2\epsilon) = 2\alpha\epsilon + 2\epsilon + 2\alpha\epsilon^2 +2\epsilon^2.
\end{align*}
Using these two bounds, we obtain
\begin{align*}
 \frac{\w_H(\pi')}{\dist_{\tilde{G}_i}(u,w)} &\le 1+(\lambda_2 - 1) \frac{p}{t}(1+\epsilon) + (\lambda_2 - 1) \frac{1}{t} \\
 &\le 1 - 1+\alpha(1+\epsilon)^3  +\alpha \epsilon (1+\epsilon)^2+3\epsilon^2+2\epsilon^3 + 2\alpha\epsilon + 2\epsilon + 2\alpha\epsilon^2 +2\epsilon^2 \\
    &\le \alpha \left((1+\epsilon)^3 + \epsilon(1+\epsilon)^2+4\epsilon+7\epsilon^2+2\epsilon^3\right)\\
    &= \alpha \left(1+3\epsilon + 3\epsilon^2 + \epsilon^3  + \epsilon + 2\epsilon^2 + \epsilon^3 +4\epsilon+7\epsilon^2+2\epsilon^3\right) \\
    &\le \alpha (1+24\epsilon).
\end{align*}
Hence, the path $\pi'$ achieves the claimed stretch of at most $\alpha(1 + 24 \epsilon)$ in all cases.
\end{proof}

By repeatedly applying the shortcuts of Lemma~\ref{lem1} and concatenating them, we construct the alternative path $\pi''$ to get the following Lemma. 
\begin{lemma}\label{lem2}
    For every pair of vertices $u,v \in V$ with $\dist_G(u,v) < \infty$, there exists a path $\pi''$ from $u$ to $v$ in the graph $H = (V', E \cup F)$, which consists of at most $(3t^2 (\alpha +1) + t)\ln(nW)$ edges, and has weight $\w_H(\pi'') \le \alpha(1+24\epsilon) \cdot \dist_G(u,v)$.
\end{lemma}

\begin{proof}
    Let $\pi$ be the shortest path between $u$ and $v$ in $G$. The goal is to find a sequence of vertices $w_0,w_1,\dots,w_k$ on $\pi$ and a sequence of alternative paths $\pi'_0, \pi'_1,\dots,\pi_{k-1}'$, where each $\pi'_i \in H$ is a path from $w_i$ to $w_{i+1}$. Let $\pi''$ denote the path obtained by concatenating all paths $\pi_i'$. Let $w_0 = u$ and let $w_i$ for $1\le i \le k$ be the vertex on $\pi$, found as described in Lemma~\ref{lem1}, with $k$ chosen such that $w_k = v$, and $\Delta$ defined as $\Delta = \frac{\dist_G(w_i,v)}{t^2}$ in order to find the vertex $w_{i+1}$. 

    From Lemma~\ref{lem1}, we know that each alternative path $\pi'_i$ has at most $3t(\alpha +1) + 1$ hops. Thus, the concatenation of all alternative paths consists of at most $k (3t(\alpha +1) + 1)$ edges.
    Likewise, Lemma~\ref{lem1} ensures that when using a shortcut $\pi_i'$ we skip a distance of $\Delta t = \frac{\dist_G(w_i,v)}{t^2}t = \frac{\dist_G(w_i,v)}{t}$, and we know that the remaining distance is:
    \begin{equation*}
        \dist_G(w_{i+1}, v) \le \dist_G(w_i, v) - \frac{\dist_G(w_i, v)}{t} = \left(1-\frac{1}{t}\right)\dist_G(w_i,v).
    \end{equation*}
    Further, we can say that for every $0\le i \le k$:
    \begin{align*} 
        \dist_G(w_i, v) \le \left(1-\frac{1}{t}\right)\dist_G(w_{i-1}, v) \le \left(1-\frac{1}{t}\right)\left(1-\frac{1}{t}\right)\dist_G(w_{i-2},v) \le \dots \le \left(1-\frac{1}{t}\right)^i \dist_G(u,v) .
    \end{align*}

    To analyze how many of such shortcuts are needed in order to reach $v$, we need to determine for which $k$ it is ensured that the remaining distance satisfies $\dist_G(w_k,v) \le (1-\frac{1}{t})^k \dist_G(u,v) < 1$.
    Since $\dist_G(u,v) \le (n-1)W < nW$, it follows that
    \begin{align*}
        1 &>\frac{1}{nW} \dist_G(u,v)  = \left(\frac{1}{e}\right)^{\ln(nW)}  \dist_G(u,v).
    \end{align*}
    Noting that $(1-\frac{1}{t})^t  \le  \frac{1}{e}$, this yields
    \begin{align*}
        \left(\frac{1}{e}\right)^{\ln(nW)} \dist_G(u,v) &\ge \left(1-\frac{1}{t}\right)^{t \cdot \ln(nW)}\dist_G(u,v).
    \end{align*}
    Thus $\dist_G(w_k,v) < 1$ holds, for every $k \ge t \cdot \ln(nW)$. Therefore the path $\pi''$ contains at most $k(3t(\alpha+1) + 1)=\ln(nW) t (t (\alpha+1) + 1) = \ln(nW)(3t^2(\alpha +1) + t)$ many edges.

    We can bound the weight of the path $\pi''$ by $\w_H(\pi'') \le \alpha (1+24\epsilon) \cdot \dist_G(u,v)$ as follows. Since Lemma~\ref{lem1} guarantees that each alternative path $\pi_i'$ has weight $\w_H(\pi'_i) \le \alpha (1+24\epsilon) \cdot \dist_G(w_i,w_{i+1})$, and all $w_i$ lie on the shortest path between $u$ and $v$ we have:
    \begin{align*}
        \w_H(\pi'') = \sum_{i = 0}^{k-1} \w_H(\pi_i') &\le \sum_{i = 0}^{k-1} \alpha (1+24\epsilon) \cdot \dist_G(w_i,w_{i+1}) \\
        &= \alpha (1+24\epsilon) \cdot \sum_{i = 0}^{k-1}\dist_G(w_i,w_{i+1}) = \alpha (1+24\epsilon) \cdot \dist_G(u,v). \qedhere
    \end{align*}
\end{proof}

\subsection{Projection to the Original Vertex Set}\label{section:projection}
So far, we have shown that the edge set $F \subseteq V' \times V'$ as constructed in Section~\ref{section:construction} is a $V$-constrained hopset of the extended graph $G' = (V',E)$.
However, our goal is to obtain a hopset for the original graph $G=(V,E)$, which consists solely of edges between original vertices $V$. 
Let $V_{\text{sub}} := V'\setminus V$ be the set of the subdivision vertices. 
In this section, we describe how to project $F$ to a hopset $F_R \subseteq V^2$ that preserves the same stretch and hopbound guarantees.
Recall Definition~\ref{def:origin} and Definition~\ref{def:delta} for a vertex in $V' = \bigcup_{i=0}^{i_{\text{max}}} \tilde{V}_i$.

Intuitively, the subdivision vertices represent specific points \textit{on} the original edges, and the respective $\delta$ values denote their exact positions relative to the endpoints of the edge.  
Therefore, our strategy with the projection is that whenever an edge in the hopset $F$ is incident to at least one vertex in $V_{\text{sub}}$, we shift it to the nearby original vertices in $V$ (illustrated in Figure~\ref{fig:projection}), while adjusting the weight using the corresponding $\delta$ value.  
This shift is realized by replacing a single edge in $F$ with up to four new edges and can be formalized as a function $\phi :F \to \mathcal{P}(V^2)$. For any edge $e = (u,v) \in F$:
\begin{equation*}
    \phi(e) :=
    \begin{cases}
        \origin(u) \times \origin(v) & \text{if }  \origin(u) \cap \origin(v) = \emptyset\\[4pt]
        \emptyset & \text{otherwise.}
    \end{cases}
\end{equation*}
\begin{figure}[H]
    \centering
        \begin{minipage}[b]{0.45\textwidth}
        \centering
        \begin{tikzpicture}[baseline={(0,0)}, scale=1.1, every node/.style={font=\small}]
            \useasboundingbox (0,-1.5) rectangle (6,1.5);
        
            \coordinate (u) at (2,0);

            \coordinate (vp) at (4,0);
            \coordinate (v) at (5,0);
            \coordinate (vpp) at (6,0);
    
            \draw [thick] (1,0) -- (u);
            \draw[thick] (vp) -- (vpp);

            \draw[red, thick] 
                (u) to [bend left= 60] 
                node [midway, above] {$(u,v) \in F$}
                (v) ;

            \draw ($(u) + (1,0)$) node[anchor=mid] {$\dots$};

            \draw[blue, thick] (u) to [bend right = 45] (vp);
            \draw[blue, thick] (u) to [bend right = 45] (vpp);

            \fill (u) circle (1.5pt) node[below=4pt] {$u$};

            \fill (vp) circle (1.5pt) node[above=4pt] {$v'$};
            \fill (vpp) circle (1.5pt) node[above=4pt] {$v''$};
            \draw[thick] ($(v)+(0,-3pt)$) -- ($(v)+(0,3pt)$) node[below=8pt] {$v$};
        \end{tikzpicture}
        \caption*{(a) $u \in V$ and $v \in V_{\text{sub}}$}
    \end{minipage}
    \hfill
    \begin{minipage}[b]{0.45\textwidth}
        \centering
        \begin{tikzpicture}[baseline={(0,0)}, scale=1.1, every node/.style={font=\small}]
            \useasboundingbox (0,-1.5) rectangle (6,1.5);
        
            \coordinate (up) at (0,0);
            \coordinate (u) at (1,0);
            \coordinate (upp) at (2,0);
            
            \coordinate (vp) at (4,0);
            \coordinate (v) at (5,0);
            \coordinate (vpp) at (6,0);
    
            \draw [thick] (up) -- (upp);
            \draw[thick] (vp) -- (vpp);

            \draw[red, thick] 
                (u) to [bend left= 60] 
                node [midway, above] {$(u,v) \in F$}
                (v) ;

            \draw ($(upp) + (1,0)$) node[anchor=mid] {$\dots$};

            \draw[blue, thick] (up) to [bend right = 45] (vp);
            \draw[blue, thick] (up) to [bend right = 45] (vpp);

            \draw[blue, thick] (upp) to [bend left = 45] (vp);
            \draw[blue, thick] (upp) to [bend right = 45] (vpp);

            \fill (up) circle (1.5pt) node[above=4pt] {$u''$};
            \fill (upp) circle (1.5pt) node[above=4pt] {$u'$};

            \fill (vp) circle (1.5pt) node[above=4pt] {$v'$};
            \fill (vpp) circle (1.5pt) node[above=4pt] {$v''$};

            \draw[thick] ($(u)+(0,-3pt)$) -- ($(u)+(0,3pt)$) node[below=8pt] {$u$};
            \draw[thick] ($(v)+(0,-3pt)$) -- ($(v)+(0,3pt)$) node[below=8pt] {$v$};
        \end{tikzpicture}
        \caption*{(b) $u,v \in  V_{\text{sub}}$}
    \end{minipage}
    \hfill
    \caption{The black lines depict the edges of $G$, original vertices in $V$ are shown as dots, whereas subdivision vertices in $V_{\text{sub}}$ are marked by a vertical line.
    The red line is a hopset edge in $F$ that shortcuts the shortest path from $u$ to $v$ in the graph $\tilde{G}_i$. The blue edges depict the edges in $\phi((u,v))$, i.e. all edges between the origins of $u$ and $v$.}
    \label{fig:projection}
\end{figure}
Before we define the weights of the projected edges, we first introduce the notion of the signed offset $\eta_{u,v}$.
\begin{definition}[Signed Offset]
    Let $e = (u,v) \in F$, and let $\pi$ be a shortest path from $u$ to $v$ in $\tilde{G}_i$.
    For $s \in \{u,v\}$ and $t \in \origin(s)$, the signed offset is defined as follows
    \begin{align*}
        \eta_{u,v}(s,t) :=  
        \begin{cases}
            - \delta(s,t) & \text{if $t$ is the first vertex in $\origin(s)$ encountered when traversing $\pi$ starting from $s$} \\
            + \delta(s,t) & \text{otherwise.}
        \end{cases}
    \end{align*} 
    Note, if $\origin(s) = \{s\}$ then $\delta(s,s) = 0$.
\end{definition}
Using this signed offset, the weight of an edge $(x,y) \in \phi(e)$ is given by
\begin{align*}
    \w_{\phi(e)}((x,y)) &:= \w_{F}((u,v)) + \eta_{u,v}(u,x) + \eta_{u,v}(v,y).
\end{align*}
We define the set of edges $F_R\subseteq V \times V$ as
\begin{equation*}
    F_{R} := \bigcup_{e\in F}\quad \phi(e).
\end{equation*} 
Finally, for each edge $f \in F_R$, its weight is defined as
\begin{equation*}
    \w_{F_R}(f) := \min(\w_{\phi(e)}(f) \mid e \in F \wedge f \in \phi(e)).
\end{equation*}

We will now show that $F_R$ forms a hopset. First, we will prove that the weights assigned to the edges in $F_R$ do not underestimate the original distances. To achieve this, we first establish two properties regarding the relationship between vertices and their origin.

\begin{observation} \label{obs:origin-contained-in-every-path} 
    Consider a distance scale $i \ge 0$.
    Let $u,v \in \tilde{V}_i$ with $\origin(u) \cap \origin(v) = \emptyset$ and let $\pi$ be any path between $u$ and $v$ in $\tilde{G}_i$.
    There exist vertices $u' \in \origin(u)$ and $v' \in \origin(v)$ that lie on $\pi$.
    Furthermore, the subpaths of $\pi$ from $u$ to $u'$ and $v$ to $v'$ are vertex disjoint, and have weights $\delta(u,u')$ and $\delta(v,v')$, respectively.
\end{observation}
\begin{proof}
    Let $\pi$ be any path from $u$ to $v$ in $\tilde{G}_i$ and let $s \in \{u,v\}$.
    
    \begin{case}
        $s \in V$. Then $\origin(s) = \{s\}$, so the subpath $\pi[s,s]$ is trivial with weight $0 = \delta(s,s)$.
    \end{case}
    \begin{case}
        $s \in \tilde{V}_i \setminus V$. Let $\origin(s) = \{s',s''\}$. By definition of $\origin(s)$ and $\tilde{G}_i$, the vertex $s$ lies on the subdivision path connecting $s'$ and $s''$, whose internal vertices all have degree two. Thus, any path from $s$ to a vertex outside this subdivision path, i.e., with a different $\origin$, must contain at least one of $s'$ or $s''$. Let $s'$ be the first of these two vertices encountered when traversing $\pi$ starting from $s$. By definition, the weight of the path from $s$ to $s'$ along this subdivision path is $\delta(s,s')$, as claimed.
    \end{case}
    Let $u'$ and $v'$ be the vertices identified for $u$ and $v$, respectively. Since $\origin(u)\cap\origin(v)=\emptyset$, the subpaths $\pi[u,u']$ and $\pi[v',v]$ are vertex disjoint.
\end{proof}

\begin{claim}\label{claim:distance-upperbound} 
    Consider a distance scale $i \ge 0$.
    For every pair of vertices $u,v \in \tilde{V}_i$ with $\dist_{\tilde{G}_i}(u,v) < \infty$ and $\origin(u) \cap\origin(v) = \emptyset$ and for any $x \in \origin(u)$ and $y \in \origin(v)$, the following holds $\dist_{\tilde{G}_i}(x,y) \le \dist_{\tilde{G}_i}(u,v) + \eta_{u,v}(u,x)+\eta_{u,v}(v,y)$. 
\end{claim}
\begin{proof}
    Let $\pi$ be a shortest path between $u$ and $v$ in $\tilde{G}_i$.
    We consider three cases.
    \begin{case}
        $u,v \in V$.
        We have  $\origin(u) = \{u\}$ and $\origin(v) = \{v\}$.
        By definition, $\delta(u,u) = \delta(v,v) = 0$ and $\eta_{u,v}(u,u) = \eta_{u,v}(v,v) = 0$. Hence,
        \begin{align*}
            \dist_{\tilde{G}_i}(u,v)= \dist_{\tilde{G}_i}(u,v) + \eta_{u,v}(u,u)+\eta_{u,v}(v,v).
        \end{align*}
    \end{case}
    \begin{case}
        $u \in V$ and $v \in \tilde{V}_i\setminus V$.
        Here $\origin(u) = \{u\}$ and let $\origin(v) = \{v',v''\}$, where $v'$ is the vertex on $\pi$ guaranteed by Observation~\ref{obs:origin-contained-in-every-path}.
        Since $\pi$ is a shortest path, we have $\delta(v,v') = \dist_{\tilde{G}_i}(v,v')$, and since $v'$ lies on $\pi$, the following holds
        \begin{align*}
            \dist_{\tilde{G}_i}(u,v') = \dist_{\tilde{G}_i}(u,v) - \dist_{\tilde{G}_i}(v,v') = \dist_{\tilde{G}_i}(u,v) - \delta(v,v')  = \dist_{\tilde{G}_i}(u,v) + \eta_{u,v}(v,v').
        \end{align*}
        For $v''$, we use the triangle inequality and the fact that $\dist_{\tilde{G}_i}(v,v'') \le \delta(v,v'')$, to obtain
        \begin{align*}
            \dist_{\tilde{G}_i}(u,v'') \le \dist_{\tilde{G}_i}(u,v)  + \dist_{\tilde{G}_i}(v,v'') \le \dist_{\tilde{G}_i}(u,v) + \delta(v,v'') = \dist_{\tilde{G}_i}(u,v) + \eta_{u,v}(v,v'').
        \end{align*}
    \end{case}
    \begin{case}
        $u,v \in \tilde{V}_i\setminus V$.
        Let $\origin(u) = \{u',u''\}$ and $\origin(v) = \{v',v''\}$, where $u'$ and $v'$ are the vertices on $\pi$ guaranteed by Observation~\ref{obs:origin-contained-in-every-path}.
        Again since $\pi$ is a shortest path, we have $\delta(u,u') = \dist_{\tilde{G}_i}(u,u')$ and $\delta(v,v') = \dist_{\tilde{G}_i}(v,v')$.
        Since $u'$ and $v'$ lie on $\pi$ and the subpaths $\pi[u,u']$ and $\pi[v',v]$ are vertex-disjoint, the distance between them is 
        \begin{align*}
            \dist_{\tilde{G}_i}(u',v') = \dist_{\tilde{G}_i}(u,v) - \dist_{\tilde{G}_i}(u,u') - \dist_{\tilde{G}_i}(v,v') = \dist_{\tilde{G}_i}(u,v) - \delta(u,u') - \delta(v,v').
        \end{align*}
        By definition of $\eta_{u,v}$, we get
        \begin{align*}
            \dist_{\tilde{G}_i}(u',v') =\dist_{\tilde{G}_i}(u,v) + \eta_{u,v}(u,u') + \eta_{u,v}(v,v').
        \end{align*}
        Similarly the distances from $u$ to $v'$ and $v$ to $u'$ are:
        \begin{equation}\label{equa:path-decomp}
            \begin{aligned}
            \dist_{\tilde{G}_i}(u,v') &= \dist_{\tilde{G}_i}(u,v) - \dist_{\tilde{G}_i}(v,v')  = \dist_{\tilde{G}_i}(u,v) - \delta(v,v') = \dist_{\tilde{G}_i}(u,v) + \eta_{u,v}(v,v') \\
            \dist_{\tilde{G}_i}(u',v) &= \dist_{\tilde{G}_i}(u,v) - \dist_{\tilde{G}_i}(u,u')   = \dist_{\tilde{G}_i}(u,v) - \delta(u,u') = \dist_{\tilde{G}_i}(u,v) + \eta_{u,v}(u,u').
            \end{aligned}
        \end{equation}
        For the remaining vertex pairs in $\origin(u) \times \origin(v)$ using the triangle inequality, yields
        \begin{align*}
            \dist_{\tilde{G}_i}(u',v'') &\le \dist_{\tilde{G}_i}(u',v) + \dist_{\tilde{G}_i}(v,v'') \\ 
            \dist_{\tilde{G}_i}(u'',v') &\le \dist_{\tilde{G}_i}(u,v') +  \dist_{\tilde{G}_i}(u,u'') \\
            \dist_{\tilde{G}_i}(u'',v'') &\le \dist_{\tilde{G}_i}(u,v)  +  \dist_{\tilde{G}_i}(u,u'') + \dist_{\tilde{G}_i}(v,v'').
        \end{align*}
        Combining Eq.~\ref{equa:path-decomp}, the bound $\delta(s,t) \ge \dist_{\tilde{G}_i}(s,t)$ for $s \in \{u,v\}$ and $t \in \origin(s)$, and the definition of $\eta_{u,v}$, we obtain
        \begin{align*}
            \dist_{\tilde{G}_i}(u',v'') &\le \dist_{\tilde{G}_i}(u',v) + \delta(v,v'') = \dist_{\tilde{G}_i}(u',v) + \eta_{u,v}(v,v'') =\dist_{\tilde{G}_i}(u,v) + \eta_{u,v}(u,u') + \eta_{u,v}(v,v'') \\ 
            \dist_{\tilde{G}_i}(u'',v') &\le \dist_{\tilde{G}_i}(u,v') +  \delta(u,u'') = \dist_{\tilde{G}_i}(u,v') +  \eta_{u,v}(u,u'') = \dist_{\tilde{G}_i}(u,v) + \eta_{u,v}(u,u'') + \eta_{u,v}(v,v')\\
            \dist_{\tilde{G}_i}(u'',v'') &\le \dist_{\tilde{G}_i}(u,v)  +  \delta(u,u'') + \delta(v,v'')  = \dist_{\tilde{G}_i}(u,v)  +  \eta_{u,v}(u,u'') + \eta_{u,v}(v,v'').
        \end{align*}
    \end{case}
    \medskip
    Therefore in all cases, the distance between $x \in \origin(u)$ and $y \in \origin(v)$ can be stated as 
    \begin{equation*}
        \dist_{\tilde{G}_i}(x,y) \le \dist_{\tilde{G}_i}(u,v) + \eta_{u,v}(u,x)+\eta_{u,v}(v,y). \qedhere
    \end{equation*}
\end{proof}

Using Claim~\ref{claim:distance-upperbound}, we now show that each edge in $F_R$ has weight at least the corresponding distance in $G$.

\begin{lemma}\label{lem:no-underestimation}
    For every edge $(x,y) \in F_R$, its assigned weight satisfies $\w_{F_R}((x,y)) \ge \dist_G(x,y)$.
\end{lemma}
\begin{proof} 
    Consider an edge $e = (u,v) \in F$ and let $(x,y) \in \phi(e)$.
    Let $i$ be any distance scale $(u,v)$ belongs to, i.e., $(u,v)$ appears in the emulator $M_i$.
    The construction of $F$ ensures that $\w_F((u,v)) \ge \dist_{\tilde{G}_i}(u,v)$ holds. 
    By definition, 
    \begin{align*}
        \w_{\phi(e)}((x,y)) &= \w_{F}((u,v)) + \eta_{u,v}(u,x) + \eta_{u,v}(v,y).
    \end{align*}
    From Claim~\ref{claim:distance-upperbound}, we know that the distance between $x$ and $y$ in $\tilde{G}_i$ satisfies
    \begin{align*}
        \dist_{\tilde{G}_i}(x,y) &\le \dist_{\tilde{G}_i}(u,v) + \eta_{u,v}(u,x) + \eta_{u,v}(v,y).
    \end{align*}
    Now it is easy to see that 
    \begin{align*}
        \w_{\phi(e)}(x,y) &= \w_{F}((u,v)) + \eta_{u,v}(u,x) + \eta_{u,v}(v,y) \\
                &\ge \dist_{\tilde{G}_i}(u,v) + \eta_{u,v}(u,x) + \eta_{u,v}(v,y) \\
                &\ge \dist_{\tilde{G}_i}(x,y) \\
                &= \dist_G(x,y).
    \end{align*}
    Since $F_R = \bigcup_{e\in F} \phi(e)$ and $\w_{F_R}((x,y))$ is the minimum over all such projections, the Lemma follows.
\end{proof}

Next, we show that $F_R$ preserves the same stretch and hopbound guarantees as $F$. In particular, we show that any path in $H = (V', E\cup F)$ can be simulated by a path in the graph $H_R = (V, E \cup F_R)$ with no increase in weight or hop-count.

\begin{lemma}\label{lem:projection-properties}
For any pair of vertices $u,v \in V$ and any path $\pi$ from $u$ to $v$ in $H = (V', E\cup F)$, there exists a path $\pi_R$ in $H_R = (V, E \cup F_R)$, such that $\w_{H_R}(\pi_R) \le \w_H(\pi)$ and $|\pi_R| \le |\pi|$.    
\end{lemma}
\begin{proof} 
    Let $\pi = (u = x_0,x_1,\dots, x_L = v)$ be a path between $u$ and $v$ in $H$.
    We construct a path $\pi_R = (u = y_0, y_1,\dots,y_{k(L)} = v)$ from $u$ to $v$ in $H_R$ iteratively.
    For each $0 \le l \le L$, let $k(l)$ denote the index of the last vertex of $\pi_R$ after processing the prefix $(x_0,\dots,x_l)$ of $\pi$.

    Let $k(0) := 0$ and $y_{k(0)} := x_0 = u$.
    Suppose that for some $0 \le l < L$, we have already constructed the sequence $(y_0,y_1,\dots,y_{\kl})$. 
    For the next vertex $x_{l+1}$ there are the following two cases:
    \begin{enumerate}[label=\arabic*.]
        \item \label{item:yk_in_origin} $y_{\kl} \in \origin(x_{l+1})$. We set $\kll := \kl$ and do not extend the sequence $(y_0,y_1,\dots,y_{\kl})$.
        \item \label{item:yk_notin_origin} $y_{\kl} \notin \origin(x_{l+1})$. In this case we set $\kll := \kl + 1$ and extend the sequence of vertices $(y_0,y_1,\dots,y_{\kl})$, with a vertex $y_{\kll} \in \origin(x_{l+1})$. Specifically:
        \begin{enumerate}[label=\arabic{enumi}.\arabic*.]
            \item If $x_{l+1} \in V$, then $\origin(x_{l+1}) = \{x_{l+1}\}$, therefore $y_{\kll} = x_{l+1}$.
            \item If $x_{l+1} \in V_{\text{sub}}$, then let $\origin(x_{l+1}) = \{x', x''\}$. We choose $y_{\kll} := x'$ , where $x'$ denotes the vertex on the shortest path from $x_l$ to $x_{l+1}$ in the subdivision graph $\tilde{G_i}$ guaranteed by Observation~\ref{obs:origin-contained-in-every-path}, where $i$ is the distance scale in which $x_{l+1}$ was created.
        \end{enumerate}
    \end{enumerate}
    Intuitively, in \ref{item:yk_in_origin} the path $\pi$ remains on the same original edge of $G$, i.e., the vertices $x_l$ and $x_{l+1}$ are both part of the same subdivision of some edge $E$. Taking the edge $(x_l,x_{l+1})$ does not move us closer to $v$ in terms of original vertices. Whereas in \ref{item:yk_notin_origin}, the edge $(x_l,x_{l+1})$, allows us to advance to $v$, by at least one original vertex of $G$. An example illustration is given in Figure~\ref{fig:projectionexample_appendix}.

    Continuing this process until $x_{l+1} = x_L = v$ yields a path $\pi_R = (y_0,y_1,\dots, y_{k(L)})$ in $H_R$. Since by construction $k(L) \le L$, we have $|\pi_R| \le |\pi|$.
    \begin{figure}[H]
        \centering
        \begin{tikzpicture}[baseline={(0,0)}], scale=1.3, every node/.style={font=\small}]
            \coordinate (u) at (0,0);
            \coordinate (x1) at (1,0);
            \coordinate (x2) at (2,0);
    
            \coordinate (xp) at (3.5,0);
            \coordinate (x3) at (4,0);
            \coordinate (x4) at (5,0);
            \coordinate (x5) at (6,0);
            \coordinate (x6) at (7,0);
            \coordinate (xpp) at (8,0);
    
            \coordinate (zp) at (9.5,0);
            \coordinate (x7) at (10,0);
            \coordinate (zpp) at (11,0);
    
            \coordinate (v) at (12.5,0);

            \foreach \i in {u, x1, x2, xp, xpp, zp, zpp, v} {
                \fill (\i) circle (1.5pt) node[above=4pt]{};
            }
    
            \foreach \i in {x3, x4, x5, x6, x7} {
                \draw[thick] ($(\i)+(0,-3pt)$) -- ($(\i)+(0,3pt)$) node[below=8pt]{};
            }
    
            \node[] at ($(u) - (0.4,0)$) {$u$};
            \node[] at ($(v) + (0.4,0)$) {$v$};
            
            \node[red, above=4pt] at (u) {$x_0$};
            \node[red, above=4pt] at (x1) {$x_1$};
            \node[red, above=4pt] at (x2) {$x_2$};
            \node[red, above=4pt] at (x3) {$x_3$};
            \node[red, above=4pt] at (x4) {$x_4$};
            \node[red, above=4pt] at (x5) {$x_5$};
            \node[red, above=4pt] at (x6) {$x_6$};
            \node[red, above=4pt] at (x7) {$x_7$};
            \node[red, above=4pt] at (v) {$x_8$};
    
            \node[blue, below=4pt] at (u) {$y_0$};
            \node[blue, below=4pt] at (x1) {$y_1$};
            \node[blue, below=4pt] at (x2) {$y_2$};
            \node[blue, below=4pt] at (xp) {$y_4$};
            \node[blue, below=4pt] at (zp) {$y_5$};
            \node[blue, below=4pt] at (v) {$y_6$};
            
            \draw[thick] (u) -- (x2);
            \draw[thick] (xp) -- (xpp);
            \draw[thick] (zp) -- (zpp);
    
    
            \draw[red, thick] (x2) to [bend left = 60] (x3);
            \draw[red, thick] (x3) to [bend left = 60] (x4);
            \draw[red, thick] (x4) to [bend left = 60] (x5);
            \draw[red, thick] (x5) to [bend left = 60] (x6);
            \draw[red, thick] (x6) to [bend left = 60] (x7);
            \draw[red, thick] (x7) to [bend left = 60] (v);
    
            \draw ($(x2) + (0.75,0)$) node[anchor=mid] {$\dots$};
            \draw ($(xpp) + (0.75,0)$) node[anchor=mid] {$\dots$};
            \draw ($(zpp) + (0.75,0)$) node[anchor=mid] {$\dots$};
    
            \draw[blue, thick] (x2) to [bend right = 45] (xp);
            \draw[blue, thick] (xp) to [bend right = 45] (zp);
            \draw[blue, thick] (zp) to [bend right = 45] (v);

        \end{tikzpicture}

        \caption{Illustration of the construction of the path $\pi_R = (y_0,y_1,\dots)$ from a given path $\pi = (x_0, x_1,\dots)$. The black edges are edges in the graph $G$. The red vertices/edges show the original path $\pi$ in $H$. The blue vertices/edges show the constructed path $\pi_R$ in $H_R$, where new vertices $y_{\kll} \in \origin(x_{l+1})$ are appended, only when $y_{\kl} \notin \origin(x_{l+1})$}
        \label{fig:projectionexample_appendix}
    \end{figure}
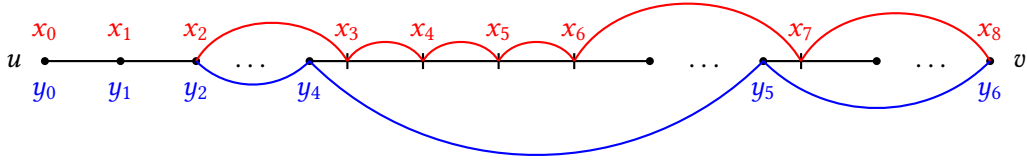
    It remains to show that the sequence of vertices $\pi_R$ is indeed a path in $H_R$ and that its weight satisfies $\w_{H_R}(\pi_R) \le \w_H(\pi)$. 
    We do this by induction on $l$: we will show, that for any $0 \le l < L$, the sequence $(y_0,\dots,y_{\kl})$ is a path in $H_R$, that $y_{\kl} \in \origin(x_l)$ and 
    \begin{equation*}
            \sum_{i = 0}^{l-1} \w_H((x_i,x_{i+1})) \ge \sum_{i = 0}^{\kl-1} \w_{H_R}((y_i,y_{i+1})) + \delta(y_{\kl}, x_l). 
    \end{equation*}
    This implies that for the last vertex, for which $x_L = v = y_{k(L)}$ and thus $\delta(y_{k(L)}, x_L) = 0$ holds, that we have
    \begin{align*}
        \w_H(\pi) = \sum_{i = 0}^{L-1} \w_H((x_i,x_{i+1})) \ge \sum_{i = 0}^{k(L)-1} \w_{H_R}((y_i,y_{i+1})) + \delta(y_{k(L)}, x_L) = \sum_{i = 0}^{K-1} \w_{H_R}((y_i,y_{i+1})) = \w_{H_R}(\pi_R).
    \end{align*}

    \medskip
    \noindent \textit{Base Case.} For $l=0$, we have $\kl=0$. The sequence $(y_0)$ is the trivial path. And since $x_0 = y_0 \in V$, $y_0 \in \origin(x_0)$ holds by definition.\\
    \noindent \textit{Induction Step.} Suppose that for some $0 \le l < L$, the sequence $(y_0,\dots,y_{\kl})$ is a path in $H_R$, that $y_{\kl} \in \origin(x_l)$ and that the following holds 
    \begin{equation}
            \sum_{i = 0}^{l-1} \w_H((x_i,x_{i+1})) \ge \sum_{i = 0}^{\kl-1} \w_{H_R}((y_i,y_{i+1})) + \delta(y_{\kl}, x_l). 
            \tag{IH}
    \end{equation}
    Now consider the next vertex $x_{l+1}$ on $\pi$, we distinguish the same cases as in the construction of $\pi_R$.

    \begin{case}
    $y_{\kl} \in \origin(x_{l+1})$. \\
    In this case the sequence $(y_0,y_1\dots ,y_{\kl})$ remained unchanged and we have $\kll = \kl$, thus $y_{\kll} \in \origin(x_{l+1})$ holds. 
    We know that 
    \[
        \w_H((x_l,x_{l+1})) \ge \delta(y_{\kl}, x_{l+1}) - \delta(y_{\kl},x_l) \tag{a}
    \]
    holds.
    This yields
    \begin{align*}
        \sum_{i = 0}^{l} \w_H((x_i,x_{i+1})) &= \sum_{i = 0}^{l-1} \w_H((x_i,x_{i+1})) + \w_H((x_l,x_{l+1})) \\
        &\mathrel{\overset{\text{(IH)}}{\ge}} \sum_{i = 0}^{\kl-1} \w_{H_R}((y_i,y_{i+1})) + \delta(y_{\kl}, x_l) + \w_H((x_l,x_{l+1})) \\
        &\mathrel{\overset{\text{(a)}}{\ge}}  \sum_{i = 0}^{\kl-1} \w_{H_R}((y_i,y_{i+1})) + \delta(y_{\kl}, x_{l+1}) \\
        &= \sum_{i = 0}^{\kll-1} \w_{H_R}((y_i,y_{i+1})) + \delta(y_{\kll}, x_{l+1}).
    \end{align*}
    where the last equality holds, since we have $\kll = \kl$.
    \end{case}
    
    \begin{case}
    $y_{\kl} \notin \origin(x_{l+1})$. \\
    According to the construction, the sequence $(y_0,y_1,\dots,y_{\kl})$ is extended by a new vertex $y_{\kll}$ and $\kll = \kl +1$. 
    \begin{subcase}
        $x_{l+1} \in V$. We chose $y_{\kll} = x_{l+1}$, since $x_{l+1} \in V$ we have $y_{\kll} \in \origin(x_{l+1})$ by definition. The edge $(y_{\kl}, x_{l+1})$ exists in $E \cup F_R$, since
        \begin{itemize}
            \item If $(y_{\kl}, x_{l+1}) \in E \cup F$, we consider the following
            \begin{itemize}
                \item If $(y_{\kl}, x_{l+1}) \in E$, then it is also in $E \cup F_R$.
                \item If $(y_{\kl}, x_{l+1}) \in F$, we have $\phi((y_{\kl}, x_{l+1})) = \{(y_{\kl}, x_{l+1})\}$ since both endpoints are in $V$. Therefore $(y_{\kl}, x_{l+1}) \in F_R$.
            \end{itemize}
            The weight of the edge is $\w_{H_R}((y_{\kl}, x_{l+1})) = \w_H((y_{\kl},x_{l+1}))  \le \delta(y_{\kl}, x_l) + \w_{H}((x_l,x_{l+1}))$.

            \item If $(y_{\kl}, x_{l+1}) \notin E \cup F$. Since $\origin(x_{l+1}) = \{x_{l+1}\}$ and $y_{\kl} \notin \origin(x_{l+1})$, we know that $\origin(x_l)\cap\origin(x_{l+1})=\emptyset$ holds. Therefore we added $\phi((x_l, x_{l+1})) = \origin(x_l) \times \origin(x_{l+1})$ to $F_R$. Since $y_{\kl} \in \origin(x_l)$ is guaranteed by the induction hypothesis, we have $(y_{\kl}, x_{l+1}) \in F_R$. Its weight is given by $\w_{H_R}((y_{\kl},x_{l+1})) \le  \w_{\phi((x_l,x_{l+1}))}((y_k,x_{l+1})) =  \delta(y_{\kl}, x_l) + \w_{H}((x_l,x_{l+1}))$.
            
        \end{itemize}
        In either case, the sequence $(y_0,y_1,\dots, y_{\kl},y_{\kll})$ forms a path in $H_R$.
        Moreover the weight of the edge $(y_{\kl},y_{\kll}) = (y_{\kl},x_{l+1})$, satisfies
        \[
            \w_{H}((x_l,x_{l+1})) \ge \w_{H_R}((y_{\kl},y_{\kll})) - \delta(y_{\kl}, x_l). \tag{b}
        \]
        Consequently, we obtain
        \begin{align*}
            \sum_{i = 0}^{l} \w_H((x_i,x_{i+1})) &= \sum_{i = 0}^{l-1} \w_H((x_i,x_{i+1})) + \w_H((x_l,x_{l+1})) \\
             &\mathrel{\overset{\text{(IH)}}{\ge}}  \sum_{i = 0}^{\kl-1} \w_{H_R}((y_i,y_{i+1})) + \delta(y_{\kl}, x_l) + \w_H((x_l,x_{l+1})) \\
            &\mathrel{\overset{\text{(b)}}{\ge}} \sum_{i = 0}^{\kl-1} \w_{H_R}((y_i,y_{i+1})) + \delta(y_{\kl}, x_l) + \w_{H_R}((y_{\kl},y_{\kll})) - \delta(y_{\kl}, x_l) \\
            &= \sum_{i = 0}^{\kl} \w_{H_R}((y_i,y_{i+1})) = \sum_{i = 0}^{\kll -1} \w_{H_R}((y_i,y_{i+1})),   
        \end{align*}
        where the final equality follows from $\kll = \kl+1$.
        Since $y_{\kll} = x_{l+1}$, we have $\delta(y_{\kll}, x_{l+1}) = 0$, which yields
        \begin{align*}
             \sum_{i = 0}^{l} \w_H((x_i,x_{i+1})) \ge \sum_{i = 0}^{\kll - 1} \w_{H_R}((y_i,y_{i+1}))  + \delta(y_{\kll}, x_{l+1}). 
        \end{align*} 
    \end{subcase}
 
    \begin{subcase}
        $x_{l+1} \in V_{\text{sub}}$. Let $\origin(x_{l+1}) = \{x', x''\}$ and let $x'$ be the vertex on the shortest path from $x_l$ to $x_{l+1}$ in the subdivision graph $\tilde{G_i}$ guaranteed by Observation~\ref{obs:origin-contained-in-every-path}, where $i$ is the distance scale in which $x_{l+1}$ was inserted. We chose $y_{\kll} = x'$, thus $y_{\kll} \in \origin(x_{l+1})$ holds.
        Since $y_{\kl} \in \origin(x_l)$ and $y_{\kl} \notin \origin(x_{l+1})$, we have $|\origin(x_l) \cap \origin(x_{l+1})| \le 1$. It follows that the edge $(y_{\kl}, x')$ exists in $E \cup F_R$, since
        \begin{itemize}
            \item If $\origin(x_l) \cap \origin(x_{l+1}) = \{x'\}$, then $(y_{\kl}, x') \in E$. \\
            Its weight is given by $\w_{H_R}((y_{\kl},y_{\kll})) = \w_{H_R}((y_{\kl},x')) = \w_G((y_{\kl},x')) \le \delta(y_{\kl},x_l) + \w_H((x_l,x_{l+1})) - \delta(x',x_{l+1})$.
    
            \item If $\origin(x_l) \cap \origin(x_{l+1}) = \emptyset$, then we added $\phi((x_l, x_{l+1})) = \origin(x_l) \times \origin(x_{l+1})$ to $F_R$. Since $y_{\kl} \in \origin(x_l)$ holds by the induction hypothesis, we have $(y_{\kl}, x_{l+1}) \in F_R$. \\
            The weight of the edge is $\w_{H_R}((y_{\kl},y_{\kll})) = \w_{H_R}((y_{\kl},x')) \le \w_{\phi((x_l,x_{l+1}))}=\delta(y_{\kl},x_l) + \w_H((x_l,x_{l+1})) - \delta(x',x_{l+1})$.
        \end{itemize}
        Thus, in both cases the sequence $(y_0,y_1,\dots, y_{\kl},y_{\kll})$ forms a path in $H_R$.
        In addition, for the weight of the edge $(y_{\kl},y_{\kll}) = (y_{\kl},x')$ the following holds
        \[
            \w_H((x_l,x_{l+1})) \ge \w_{H_R}((y_{\kl},y_{\kll})) - \delta(y_{\kl},x_l) + \delta(x',x_{l+1}) \tag{c}
        \]
        Hence, we have
        \begin{align*}
            \sum_{i = 0}^{l} \w_H((x_i,x_{i+1})) &= \sum_{i = 0}^{l-1} \w_H((x_i,x_{i+1})) + \w_H((x_l,x_{l+1})) \\
            &\mathrel{\overset{\text{(IH)}}{\ge}}  \sum_{i = 0}^{\kl-1} \w_{H_R}((y_i,y_{i+1})) + \delta(y_{\kl}, x_l) + \w_H((x_l,x_{l+1})) \\
            &\mathrel{\overset{\text{(c)}}{\ge}} \sum_{i = 0}^{\kl-1} \w_{H_R}((y_i,y_{i+1})) + \delta(y_{\kl}, x_l) + w_{H_R}((y_{\kl},y_{\kll})) - \delta(y_{\kl},x_l) + \delta(x',x_{l+1}) \\
            &= \sum_{i = 0}^{\kl-1} \w_{H_R}((y_i,y_{i+1})) + w_{H_R}((y_{\kl},y_{\kll})) + \delta(x',x_{l+1}) \\
            &= \sum_{i = 0}^{\kl} \w_{H_R}((y_i,y_{i+1}))  + \delta(y_{\kll}, x_{l+1})\\ 
            &= \sum_{i = 0}^{\kll - 1} \w_{H_R}((y_i,y_{i+1}))  + \delta(y_{\kll}, x_{l+1}),
        \end{align*}
        where the final step follows from $\kll = \kl + 1$. \qedhere  
    \end{subcase}    
    \end{case}    
\end{proof}

It remains to show that the hopset $F_R$ inherits the size of the underlying emulator construction.

\begin{lemma}\label{lem:size-projection} 
    Let $\mathcal{A}$ be an algorithm that, for any $n$-vertex graph, constructs an $(\alpha,\beta)$-emulator of size $S_{\mathcal{A}}(n,\alpha,\beta)$.  
    Let $F \subseteq V'\times V'$ be the hopset produced by the construction of Section~\ref{section:construction} using $\mathcal{A}$, and let $F_R$ be the hopset obtained by projection $F$ onto the vertex set $V$. Then the size of $F_R$ is bounded by $\bigO(S_{\mathcal{A}}(n+m\frac{t}{\epsilon},\alpha,\beta)\log_{1+\epsilon}(nW))$.
\end{lemma}

\begin{proof}
    By Lemma~\ref{lem:sizeHopset} the hopset $F$ consists of $\bigO(S_{\mathcal{A}}(n+m\frac{t}{\epsilon},\alpha,\beta)\log_{1+\epsilon}(nW) + m\frac{t}{\epsilon}\log_{1+\epsilon}(\frac{t}{\epsilon}))$ many edges. As mentioned in the proof of Lemma~\ref{lem:sizeHopset}, the second additive term   $\bigO(m\frac{t}{\epsilon}\log_{1+\epsilon}(\frac{t}{\epsilon}))$ arises from the auxiliary edge set $\tilde{F}$, which contains all edges that connect each vertex $v \in V_{\text{sub}}$ to the vertices in $\origin(v)$. Hence, by definition all edges in $(u,v) \in \tilde{F}$ satisfy $\origin(u) \cap \origin(v) \neq \emptyset$, and are therefore not included in the projected hopset $F_R$.
    Furthermore all remaining edges in $F\setminus \tilde{F}$ are mapped to at most four edges in $F_R$, and thus the size of $F_R$ is bounded by $\bigO(S_{\mathcal{A}}(n+m\frac{t}{\epsilon},\alpha,\beta)\log_{1+\epsilon}(nW))$ as claimed.
\end{proof}

Having established the necessary lemmas, we are now ready to prove Theorem~\ref{thm:main}.
\thmmain*

\begin{proof}
    Let $F$ be the hopset produced by the construction in Section~\ref{section:construction}, with $t = \max(\frac{1}{\epsilon},\frac{\beta}{\epsilon})$, and let $F_R$ be its projection onto the vertex set $V$ as described in Section~\ref{section:projection}.
    By Lemma~\ref{lem:size-projection}, the number of edges in $F_R$ is bounded by $\bigO(S_{\mathcal{A}}(n+m\frac{t}{\epsilon},\alpha,\beta)\log_{1+\epsilon}(nW)) = \bigO(S_{\mathcal{A}}(n+m\frac{t}{\epsilon},\alpha,\beta)\frac{1}{\epsilon}\ln(nW))$ for $0 < \epsilon \le 1$.
    As shown in Lemma~\ref{lem2}, $F$ is a $(\alpha(1+\epsilon), \bigO(\alpha t^2)\ln(nW))$-hopset of $G'$ with respect to $V$.
    Lemma~\ref{lem:projection-properties} ensures that $F_R$ preserves the same stretch and hopbound as $F$, and Lemma~\ref{lem:no-underestimation} ensures that the distances are not underestimated.
    Therefore, $F_R$ is a $(1+\epsilon, \bigO(\alpha t^2)\ln(nW))$-hopset of $G$ with size $\bigO(S_{\mathcal{A}}(n+m\frac{t}{\epsilon},\alpha,\beta)\log_{1+\epsilon}(nW))$ .
\end{proof}

\printbibliography[heading=bibintoc] 

@inproceedings{EN17,
author = {Michael Elkin and Ofer Neiman},
title = {Efficient Algorithms for Constructing Very Sparse Spanners and Emulators},
booktitle = {Proceedings of the 2017 Annual ACM-SIAM Symposium on Discrete Algorithms (SODA)},
year = {2017},
pages = {652--669},
doi = {10.1137/1.9781611974782.41}
}

@article{HKN18,
    author = {Henzinger, Monika and Krinninger, Sebastian and Nanongkai, Danupon},
    title = {Decremental Single-Source Shortest Paths on Undirected Graphs in Near-Linear Total Update Time},
    year = {2018},
    issue_date = {December 2018},
    publisher = {Association for Computing Machinery},
    address = {New York, NY, USA},
    volume = {65},
    number = {6},
    issn = {0004-5411},
    doi = {10.1145/3218657},
    journal = {J. ACM},
    month = nov,
    articleno = {36},
    numpages = {40}
}

@inproceedings{HKN16, 
   series={STOC ’16},
   title={A deterministic almost-tight distributed algorithm for approximating single-source shortest paths},
   DOI={10.1145/2897518.2897638},
   booktitle={Proceedings of the forty-eighth annual ACM symposium on Theory of Computing},
   publisher={ACM},
   author={Monika Henzinger and Sebastian Krinninger and Danupon Nanongkai},
   year={2016},
   month=jun, pages={489–498},
   collection={STOC ’16} 
}

@article{EP04,
  title={(1+$\epsilon$,$\beta$)-spanner constructions for general graphs},
  author={Elkin, Michael and Peleg, David},
  journal={SIAM Journal on Computing},
  volume={33},
  number={3},
  pages={608--631},
  year={2004},
  publisher={SIAM}
}

@inproceedings{Coh94,
  title={Polylog-time and near-linear work approximation scheme for undirected shortest paths},
  author={Cohen, Edith},
  booktitle={Proceedings of the twenty-sixth annual ACM symposium on Theory of Computing},
  pages={16--26},
  year={1994}
}

@inproceedings{TZ06,
  title={Spanners and emulators with sublinear distance errors},
  author={Thorup, Mikkel and Zwick, Uri},
  booktitle={Proceedings of the seventeenth annual ACM-SIAM symposium on Discrete algorithm},
  pages={802--809},
  year={2006}
}

@article{HP19,
  title={Thorup--zwick emulators are universally optimal hopsets},
  author={Huang, Shang-En and Pettie, Seth},
  journal={Information Processing Letters},
  volume={142},
  pages={9--13},
  year={2019},
  publisher={Elsevier}
}

@article{PS89,
  title={Graph spanners},
  author={Peleg, David and Sch{\"a}ffer, Alejandro A},
  journal={Journal of graph theory},
  volume={13},
  number={1},
  pages={99--116},
  year={1989},
  publisher={Wiley Online Library}
}

@article{EN19,
author = {Elkin, Michael and Neiman, Ofer},
title = {Hopsets with Constant Hopbound, and Applications to Approximate Shortest Paths},
journal = {SIAM Journal on Computing},
volume = {48},
number = {4},
pages = {1436-1480},
year = {2019},
doi = {10.1137/18M1166791}}

@article{EN20,
  author       = {Michael Elkin and
                  Ofer Neiman},
  title        = {Near-Additive Spanners and Near-Exact Hopsets, {A} Unified View},
  journal      = {Bull. {EATCS}},
  volume       = {130},
  year         = {2020}}

@inproceedings{Ber09,
  title={Fully dynamic (2+ $\varepsilon$) approximate all-pairs shortest paths with fast query and close to linear update time},
  author={Bernstein, Aaron},
  booktitle={2009 50th Annual IEEE Symposium on Foundations of Computer Science},
  pages={693--702},
  year={2009},
  organization={IEEE}
}

@inproceedings{GW20,
  title={Deterministic algorithms for decremental approximate shortest paths: Faster and simpler},
  author={Gutenberg, Maximilian Probst and Wulff-Nilsen, Christian},
  booktitle={Proceedings of the Fourteenth Annual ACM-SIAM Symposium on Discrete Algorithms},
  pages={2522--2541},
  year={2020},
  organization={SIAM}
}

@inproceedings{MVPX15,
  title={Improved parallel algorithms for spanners and hopsets},
  author={Miller, Gary L and Peng, Richard and Vladu, Adrian and Xu, Shen Chen},
  booktitle={Proceedings of the 27th ACM Symposium on Parallelism in Algorithms and Architectures},
  pages={192--201},
  year={2015}
}

@article{ABP18,
  title={A hierarchy of lower bounds for sublinear additive spanners},
  author={Abboud, Amir and Bodwin, Greg and Pettie, Seth},
  journal={SIAM Journal on Computing},
  volume={47},
  number={6},
  pages={2203--2236},
  year={2018},
  publisher={SIAM}
}

@inproceedings{KP22a,
  title={Having hope in hops: New spanners, preservers and lower bounds for hopsets},
  author={Kogan, Shimon and Parter, Merav},
  booktitle={2022 IEEE 63rd Annual Symposium on Foundations of Computer Science (FOCS)},
  pages={766--777},
  year={2022},
  organization={IEEE}
}

@inproceedings{KP25,
  title={Having Hope in Missing Spanners: New Distance Preservers and Light Hopsets},
  author={Kogan, Shimon and Parter, Merav},
  booktitle={Proceedings of the 2025 Annual ACM-SIAM Symposium on Discrete Algorithms (SODA)},
  pages={4352--4374},
  year={2025},
  organization={SIAM}
}

@inproceedings{NS22,
  title={A unified framework for hopsets},
  author={Neiman, Ofer and Shabat, Idan},
  booktitle={30th Annual European Symposium on Algorithms (ESA 2022)},
  pages={81--1},
  year={2022},
  organization={Schloss Dagstuhl--Leibniz-Zentrum f{\"u}r Informatik}
}

@inproceedings{CFR20a,
  title={Efficient construction of directed hopsets and parallel approximate shortest paths},
  author={Cao, Nairen and Fineman, Jeremy T and Russell, Katina},
  booktitle={Proceedings of the 52nd Annual ACM SIGACT Symposium on Theory of Computing},
  pages={336--349},
  year={2020}
}

@inproceedings{CFR20b,
  title={Improved work span tradeoff for single source reachability and approximate shortest paths},
  author={Cao, Nairen and Fineman, Jeremy T and Russell, Katina},
  booktitle={Proceedings of the 32nd ACM Symposium on Parallelism in Algorithms and Architectures},
  pages={511--513},
  year={2020}
}

@inproceedings{CF23,
  title={Parallel exact shortest paths in almost linear work and square root depth},
  author={Cao, Nairen and Fineman, Jeremy T},
  booktitle={Proceedings of the 2023 Annual ACM-SIAM Symposium on Discrete Algorithms (SODA)},
  pages={4354--4372},
  year={2023},
  organization={SIAM}
}

@inproceedings{LN22,
  author       = {Jakub {\L}{\k{a}}cki and
                  Yasamin Nazari},
  title        = {Near-Optimal Decremental Hopsets with Applications},
  booktitle    = {49th International Colloquium on Automata, Languages, and Programming,
                  {ICALP} 2022, Paris, France, July 4-8, 2022},
  series       = {LIPIcs},
  pages        = {86:1--86:20},
  publisher    = {Schloss Dagstuhl - Leibniz-Zentrum f{\"{u}}r Informatik},
  year         = {2022}
}

@inproceedings{Fin18,
  title={Nearly work-efficient parallel algorithm for digraph reachability},
  author={Fineman, Jeremy T},
  booktitle={Proceedings of the 50th Annual ACM SIGACT Symposium on Theory of Computing},
  pages={457--470},
  year={2018}
}

@article{CDKL21,
  author       = {Keren Censor{-}Hillel and
                  Michal Dory and
                  Janne H. Korhonen and
                  Dean Leitersdorf},
  title        = {Fast approximate shortest paths in the congested clique},
  journal      = {Distributed Comput.},
  volume       = {34},
  number       = {6},
  pages        = {463--487},
  year         = {2021},
  doi          = {10.1007/s00446-020-00380-5}
}

@inproceedings{BW23,
  title={Closing the gap between directed hopsets and shortcut sets},
  author={Bernstein, Aaron and Wein, Nicole},
  booktitle={Proceedings of the 2023 Annual ACM-SIAM Symposium on Discrete Algorithms (SODA)},
  pages={163--182},
  year={2023},
  organization={SIAM}
}

@inproceedings{BH23,
  title={Folklore sampling is optimal for exact hopsets: Confirming the $\sqrt{n}$ barrier},
  author={Bodwin, Greg and Hoppenworth, Gary},
  booktitle={2023 IEEE 64th Annual Symposium on Foundations of Computer Science (FOCS)},
  pages={701--720},
  year={2023},
  organization={IEEE}
}

@inproceedings{KP22b,
author = {Shimon Kogan and Merav Parter},
title = {New Diameter-Reducing Shortcuts and Directed Hopsets: Breaking the $\sqrt{n}$ Barrier},
booktitle = {Proceedings of the 2022 Annual ACM-SIAM Symposium on Discrete Algorithms (SODA)},
chapter = {},
pages = {1326-1341},
year = {2022},
doi = {10.1137/1.9781611977073.55}
}

@inproceedings{EN19b,
  author = {Elkin, Michael and Neiman, Ofer},
  title = {Linear-Size Hopsets with Small Hopbound, and Constant-Hopbound Hopsets in RNC},
  booktitle = {The 31st {ACM} on Symposium on Parallelism in Algorithms and Architectures, {SPAA} 2019, Phoenix, AZ, USA, June 22-24, 2019},
  year = {2019},
  isbn = {9781450361842},
  publisher = {Association for Computing Machinery},
  address = {New York, NY, USA},
  doi = {10.1145/3323165.3323177},
  pages = {333–341},
  numpages = {9},
  location = {Phoenix, AZ, USA},
  series = {SPAA '19}
}

@inproceedings{BenLevyP20,
  author = {Uri Ben-Levy and Merav Parter},
  title = {New ($\alpha, \beta$) Spanners and Hopsets},
  booktitle = {Proceedings of the 2020 ACM-SIAM Symposium on Discrete Algorithms (SODA)},
  pages = {1695-1714},
  publisher = {{SIAM}},
  year = {2020},
  doi = {10.1137/1.9781611975994.104}}

@article{DHZ00,
  author       = {Dorit Dor and
                  Shay Halperin and
                  Uri Zwick},
  title        = {All-Pairs Almost Shortest Paths},
  journal      = {{SIAM} J. Comput.},
  volume       = {29},
  number       = {5},
  pages        = {1740--1759},
  year         = {2000},
  url          = {https://doi.org/10.1137/S0097539797327908},
  doi          = {10.1137/S0097539797327908}
}

\end{document}